\documentclass[11pt, a4paper]{article}

\usepackage{authblk}

\usepackage{amsmath,amssymb,amsfonts,mathtools,amsthm}    
\usepackage{graphicx,subcaption}   
\usepackage{cite}       
\usepackage{hyperref}   
\usepackage{booktabs} 
\usepackage{siunitx}
\usepackage[numbers]{natbib}
\usepackage{enumitem}
\usepackage[title]{appendix}
\usepackage[normalem]{ulem} 

\RequirePackage{geometry}
\DeclareMathOperator*{\argopt}{arg\;opt}
\DeclareMathOperator*{\argmax}{arg\;max}
\usepackage{xcolor}

\newtheorem{theorem}{Theorem}
\newtheorem{definition}{Definition}
\newtheorem{lemma}{Lemma}
\newtheorem{notation}{Notation}
\newtheorem{conjecture}{Conjecture}

\allowdisplaybreaks[3]

\title{Performance Upper Bound of Grover-Mixer Quantum Alternating Operator Ansatz}

\author[1]{Ningyi Xie\thanks{E-mail: \href{mailto:nyxie@cavelab.cs.tsukuba.ac.jp}{nyxie@cavelab.cs.tsukuba.ac.jp}}}
\author[1]{Jiahua Xu}
\author[2]{Tiejin Chen}
\author[3]{Xinwei Lee\thanks{Currently affiliated with Singapore Management University}}
\author[4]{Yoshiyuki Saito}
\author[5]{Nobuyoshi Asai}
\author[6]{Dongsheng Cai}

\affil[1]{Graduate School of Science and Technology, University of Tsukuba}
\affil[2]{School of Computing and Augmented Intelligence, Arizona State University}
\affil[3]{Graduate School of Systems and Information Engineering, University of Tsukuba}
\affil[4]{Graduate School of Computer Science and Engineering, University of Aizu}
\affil[5]{School of Computer Science and Engineering, University of Aizu}
\affil[6]{Faculty of Engineering, Information and Systems, University of Tsukuba}

\date{}

\begin{document}

\maketitle

\begin{abstract}
    The Quantum Alternating Operator Ansatz (QAOA) represents a branch of quantum algorithms for solving combinatorial optimization problems. A specific variant, the Grover-Mixer Quantum Alternating Operator Ansatz (GM-QAOA), ensures uniform amplitude across states that share equivalent objective values. This property makes the algorithm independent of the problem structure, focusing instead on the distribution of objective values within the problem. In this work, we prove the probability upper bound for measuring a computational basis state from a GM-QAOA circuit with a given depth, which is a critical factor in QAOA cost. Using this, we derive the upper bounds for the probability of sampling an optimal solution, and for the approximation ratio of maximum optimization problems, both dependent on the objective value distribution. Through numerical analysis, we link the distribution to the problem size and build the regression models that relate the problem size, QAOA depth, and performance upper bound. Our results suggest that the GM-QAOA provides a quadratic enhancement in sampling probability and requires circuit depth that scales exponentially with problem size to maintain consistent performance.
\end{abstract}

\section{Introduction}
Combinatorial optimization problems are continuously studied by both industry and academia due to their broad applicability and inherent complexity. 
Since the number of possible solutions grows exponentially with the problem size, the computational cost to find the exact optimal solution skyrockets. This motivates the development of heuristic methods to find approximate solutions in a reasonable time \cite{pearl1984heuristics,papadimitriou1998combinatorial}. In recent years, this area of research has also been activated in the context of quantum computing \cite{hadfield2019quantum,rajak2023quantum,skolik2023equivariant,cerezo2021variational}, driven by the potential for speedups that mirror achievements in other domains \cite{grover1996fast,shor1999polynomial,montanaro2016quantum}. 

A popular family of quantum algorithms for addressing combinatorial optimization problems is the Quantum Alternating Operator Ansatz (QAOA) \cite{farhi2014quantum,hadfield2019quantum,wang2020x,Bartschi_2020,herrman2022multi, Vijendran_2024}. Inspired by the principles of adiabatic quantum computing \cite{farhi2000quantum}, QAOA algorithms prepare a parameterized quantum state through an alternating sequence of operations repeated for a pre-defined number of rounds $p$, known as the circuit depth. Many works have theoretically and numerically analyzed the effect of circuit depth on the quality of solutions obtained from QAOAs \cite{zhou2020quantum,mcclean2021low,Akshay_2022,shaydulin2023evidence}, including Grover-Mixer Quantum Alternating Operator Ansatz (GM-QAOA) \cite{Bartschi_2020,golden2023numerical,benchasattabuse2023lower}.

The GM-QAOA \cite{Bartschi_2020} initializes the quantum circuit in a uniform amplitude superposition of states encoding all possible solutions in the search space, and the search space is maintained by the Grover mixer. This allows GM-QAOA to inherently preserve the feasibility of constrained problems like the traveling salesman problem and the capacitated vehicle routing problem \cite{dantzig1959truck}, given efficient preparation of the feasible state superposition \cite{Bartschi_2020,xie2024feasibilitypreserved}. In addition, GM-QAOA assigns equal amplitudes to states corresponding to the same objective value, making the algorithm's performance solely determined by the distribution of objectives.  
Though GM-QAOA exhibits a better performance on small-scale problems \cite{golden2023quantum,xie2024feasibilitypreserved,zhang2024grover} compared to standard QAOA \cite{farhi2014quantum}, the performance scalability remains unclear. 
\citet{benchasattabuse2023lower} analyzes the least circuit depth $p$ required by GM-QAOA to achieve a targeted performance, establishing theoretical bounds by extending the theorem on quantum annealing time \cite{PhysRevLett.130.140601}. However, this depth scaling may not be tight, as numerical experiments in \cite{golden2023numerical} suggest conflicting results, showing the exponential growth of the required GM-QAOA depth as the problem size increases. 
Therefore, our work aims to provide a tight upper bound for GM-QAOA, which facilitates further analysis of the resource requirements of the GM-QAOA.
Recently, \citet{bridi2024analytical} also established the upper bounds for the GM-QAOA. Their method relies on a contradiction argument based on Grover’s search optimality \cite{hamann2021quantum}, while our approach employs a mathematical optimization approach using a relaxation function to find critical points through partial derivatives to determine the upper bound.

In this work, we derive upper bounds on the performance of GM-QAOA in terms of two key metrics: the probability of sampling the optimal solution and the approximation ratio. 
These bounds stem from our proof of the upper limit on the probability of measuring a computational basis state from the GM-QAOA circuit. 
We validate the performance bounds by comparing them to the empirical values sampled from the optimized GM-QAOA states. Unlike the previous bounds derived in \cite{benchasattabuse2023lower}, our bounds exhibit a consistent alignment with the simulation results of the GM-QAOA. 
Building upon the validated performance upper bounds, we explore the scalability of GM-QAOA. Through numerical analysis of the objective value distributions, we establish the relationship between problem size, QAOA depth, and performance upper bound for several widely studied combinatorial optimization problems, including traveling salesman problem, max-$k$-colorable-subgraph, max-cut, max-$k$-vertex-cover, with the problem definitions and instance sets detailed in Appendix~\ref{sec:problemset}. The results provide further evidence of the exponential resource requirements of GM-QAOA as the problem size increases.

The contributions of this study are summarized as follows:
    \begin{enumerate}[label=(\arabic*)]
    \item \textbf{Upper Bound Proofs:} For a fixed-depth GM-QAOA, we prove:
    \begin{enumerate}[label=\alph*)]
        \item the upper bound for the probability of measuring an arbitrary computational basis state,
        \item the upper bound for the probability of sampling the optimal solution, and
        \item the upper bound for the approximation ratio.
    \end{enumerate}
    \item \textbf{Scalability Evaluation:} Utilizing the proved bounds, we demonstrate the scalability issues of GM-QAOA, through numerical analysis of the objective value distributions for problems of varying sizes.
    \end{enumerate}





\section{Background of the GM-QAOA}
A combinatorial optimization problem is defined by $(C,F)$, where search space $F$ represents the finite set of possible solutions, and $C: F \rightarrow \mathbb{R}$ is the objective function that assigns a numerical value to each solution in $F$. The goal is to find an optimal solution that maximizes or minimizes the objective function, mathematically formulated as,
\begin{equation}
    \argopt_{f\in F} C(f).
\end{equation}

For a given problem $(C,F)$, 
the QAOA approaches yield solutions by performing measurements in the computational basis of parameterized quantum circuits, $\vert \psi(\boldsymbol{\theta})\rangle$. 
Here, the parameter $\boldsymbol{\theta}$ is typically tuned by optimizing the expected value $E(\boldsymbol{\theta})$, 
\begin{equation}
    E(\boldsymbol{\theta}) \coloneqq \langle \psi(\boldsymbol{\theta})\vert H_C \vert \psi(\boldsymbol{\theta})\rangle,
    \label{eqn:expected}
\end{equation}
where the problem Hamiltonian, $H_C$, satisfies $\langle f \vert H_C \vert f\rangle=C(f)$, $\forall f \in F$. 
Specifically, $\vert f \rangle$ represents the computational basis state encoding the solution $f$.

The GM-QAOA \cite{Bartschi_2020} prepares a parameterized state $\vert\psi_{p,\mathcal{C},F}(\boldsymbol{\gamma},\boldsymbol{\beta})\rangle$ from a uniform amplitude superposition, denoted as $\vert F \rangle$: 
\begin{equation}
    \vert F \rangle \coloneqq \frac{1}{\sqrt{\vert F \vert}}\sum_{f \in F}\vert f \rangle.
\end{equation}
Then, the state evolves through $p$ repetitions of two distinctive types of operation, $U^{(P)}_{\mathcal{C}}$ and $U^{(M)}_F$, 
which is given by: 
\begin{equation}
    \vert \psi_{p,\mathcal{C},F}(\boldsymbol{\gamma},\boldsymbol{\beta})\rangle \coloneqq U^{(M)}_F(\beta_p) U^{(P)}_{\mathcal{C}}(\gamma_p) \cdots U^{(M)}_F(\beta_2) U^{(P)}_{\mathcal{C}}(\gamma_2) U^{(M)}_F(\beta_1) U^{(P)}_{\mathcal{C}}(\gamma_1)\vert F \rangle,
\end{equation}
where $p$ denotes the pre-defined QAOA depth and $\boldsymbol{\gamma} = [\gamma_1, \gamma_2, \cdots,\gamma_p]^T $, $\boldsymbol{\beta} = [\beta_1, \beta_2, \cdots,\beta_p]^T$ are tunable circuit parameters.
The phase separation operation, $U^{(P)}_{\mathcal{C}}$, functions as,
\begin{equation}
    U^{(P)}_{\mathcal{C}}(\gamma) \vert f \rangle = e^{-i\gamma \mathcal{C}(f)} \vert f \rangle.
    \label{eqn:UP}
\end{equation}
Here, we denote $\mathcal{C}: F \rightarrow \mathbb{R}$ as the “phase function”.
Generally, it is pre-defined as $\mathcal{C}(f) \coloneqq C(f)$, and $U^{(P)}_{\mathcal{C}} \coloneqq e^{-i\gamma H_C}$.
\citet{golden2021threshold} introduce a threshold-based strategy as follows: 
\begin{equation}
    \mathcal{C}(f) = 
    \begin{cases}
        1 & C(f) \leqslant th \\
        0 & otherwise
    \end{cases}
\end{equation}
where threshold, $th$, is an additional tunable parameter. 
The mixing operation $U^{(M)}_F$ in the GM-QAOA is defined as: 
\begin{equation}
    \begin{split}
        U^{(M)}_F(\beta) & \coloneqq e^{-i\beta \vert F \rangle \langle F \vert} \\
                         & = I - (1-e^{-i\beta})\vert F \rangle \langle F \vert,
    \end{split}
    \label{eqn:UM}
\end{equation}
which has a Grover-like form \cite{grover1996fast}.

We adopt two metrics to evaluate the effectiveness of states prepared by the GM-QAOA:
\begin{definition} [Probability of sampling the optimal solution]
    Given a problem $(C,F)$, 
    let $F^\ast$ denote the set of optimal solutions, where 
    $F^\ast \coloneqq \{f^\ast|f^\ast=\argopt_{f\in F} C(f)\}$.
    Let $\vert \psi(\boldsymbol{\theta})\rangle$ be the prepared quantum state (i.e., $\vert \psi_{p,\mathcal{C},F}(\boldsymbol{\gamma},\boldsymbol{\beta})\rangle$ in the GM-QAOA), 
    then the probability of sampling the optimal solution denoted $\lambda$,
    is defined as: 
    \begin{equation}
        \begin{split}
            \lambda \coloneqq \sum_{f^\ast \in F^\ast}\vert \langle f^\ast \vert \psi(\boldsymbol{\theta})\rangle \vert^2;
        \end{split}
    \end{equation}
\end{definition}
\begin{definition} [Approximation ratio]
    \label{def:approx}
    Given a maximum optimization problem $(C,F)$, 
    and let $H_C$ denote the problem Hamiltonian. 
    the approximation ratio of a prepared quantum state $\vert \psi(\boldsymbol{\theta})\rangle$ (i.e., $\vert \psi_{p,\mathcal{C},F}(\boldsymbol{\gamma},\boldsymbol{\beta})\rangle$ in the GM-QAOA), 
    denoted $\alpha$, is defined as: 
    \begin{equation}
        \alpha \coloneqq \frac{\langle \psi(\boldsymbol{\theta})\vert H_C \vert \psi(\boldsymbol{\theta})\rangle}{\max_{f\in F}C(f)}.
        \label{eqn:approx}
    \end{equation}
\end{definition}

The probability of sampling an optimal solution, $\lambda$, directly influences the QAOA time-to-solution (TTS), which is defined as $\frac{1}{\lambda}$ \cite{shaydulin2023evidence}. This represents the expected number of measurements required to sample an optimal solution from the QAOA state.  Additionally, the approximation ratio, $\alpha$, is widely adopted to evaluate approximation algorithms. It provides a measure of how close the solution given by the QAOA is to the optimal solution. This work establishes the upper bounds for both metrics of the constant depth GM-QAOA circuits.

\section{Results}
\subsection{Performance Upper Bound}
In this section, we first introduce theoretical upper limits on the probability of measuring a computational basis state from GM-QAOA circuits. 
\begin{theorem}
    \label{the:1}
    Given a problem defined with the search space $F$, using any phase function $\mathcal{C}: F \rightarrow \mathbb{R}$,
    the probability of sampling an arbitrary computational basis state $\vert f \rangle$ from a depth-$p$ GM-QAOA circuit has an upper bound: 
    \begin{equation}
        \vert \langle f \vert \psi_{p,\mathcal{C},F}(\boldsymbol{\gamma},\boldsymbol{\beta})\rangle \vert^2 < \frac{(2p+1)^2}{\vert F \vert}.
    \end{equation}
\end{theorem}
\begin{proof}
    See Appendix~\ref{sec:profThe1} for details and the proof sketch is as follows. The proof begins by expanding the probability $\vert \langle f \vert \psi_{p,\mathcal{C},F}(\boldsymbol{\gamma},\boldsymbol{\beta})\rangle \vert^2$ and introducing a relaxation function $G$ by including new variables, allowing the same or greater range. The search for the maximum is then narrowed to the set of points where the partial derivatives of $G$ are zero. This set is further reduced by identifying subsets yielding identical $G$ outputs. Finally, the maximum value of $G$ is found within this restricted search space, thereby establishing the upper bound of the probability $\vert \langle f \vert \psi_{p,\mathcal{C},F}(\boldsymbol{\gamma},\boldsymbol{\beta})\rangle \vert^2$. 
\end{proof}

Building on this fundamental result, we further derive upper bounds for the probability of sampling the optimal solution and the approximation ratio achieved by GM-QAOA, each based on a statistical metric that assesses the distribution of objective function values, defined as follows.

\begin{definition}[Optimality density]
    \label{def:od}
    Given a problem $(C,F)$, 
    and let $F^\ast$ denote the set of optimal solutions, where 
    $F^\ast \coloneqq \{f^\ast|f^\ast=\argopt_{f\in F} C(f)\}$.
    Then, the optimality density, denoted $\rho $,
    is defined as, 
    \begin{equation}
        \begin{split}
            \rho \coloneqq \frac{\vert F^\ast \vert}{\vert F \vert}.
        \end{split}
    \end{equation}
\end{definition}
The optimality density $\rho$ represents the proportion of optimal solutions among all solutions in the search space. Using $\rho$, from Theorem~\ref{the:1}, we can straightforwardly determine the upper bound for the probability of sampling the optimal solution. This leads us to the following theorems:
\begin{theorem}
    \label{the:2}
    Given a problem $(C,F)$, 
    where the optimality density of the distribution of objective values is $\rho$.
    Then, the probability of sampling the optimal solution from a depth-$p$ GM-QAOA circuit is bounded as,
    \begin{equation}
        \lambda < (2p+1)^2\rho.
    \end{equation}
\end{theorem}
\begin{proof}
    By applying the probability upper bound from Theorem~\ref{the:1} to each optimal solution and summing over, then, we get the upper bound for the probability of sampling the optimal solution.
\end{proof}

Theorem \ref{the:2} reveals a significant limitation of the GM-QAOA for sampling optimal solution. In typical combinatorial optimization problems, the search space grows exponentially with problem size, while the number of optimal solutions remains relatively small or grows at a much slower rate. Consequently, the optimality density $\rho$ decreases exponentially with increasing problem size. Despite the quadratic enhancement factor $(2p+1)^2$ provided by GM-QAOA, this is insufficient to counteract the exponential decrease in  $\rho$. We further numerically demonstrate this in Section \ref{sec:lambda}.

\begin{definition}[Top-$r$-proportion mean-max ratio]
    \label{def:mmr}
    Given a problem $(C,F)$, 
    sort the solutions as $f^{(1)}, f^{(2)},\cdots, f^{(\vert F \vert)}$, 
    based on their objective function values,
    such that 
    \begin{equation}
        C(f^{(1)})\geqslant C(f^{(2)})\geqslant \cdots \geqslant C(f^{(\vert F \vert)}).
        \label{eqn:sort}
    \end{equation}
    The top-$r$-proportion mean-max ratio, $\mu_r$, 
    measures the mean value of top $r$ proportion objective values over the maximum value, defined as: 
    \begin{equation}
        \mu_r \coloneqq \frac{\sum_{i=1}^{\left\lceil r\vert F \vert \right\rceil } C(f^{(i)})}{r\vert F \vert C(f^{(1)})}.
    \end{equation}
\end{definition}
The top-$r$-proportion mean-max ratio $\mu_r$ suggests an approximation ratio where probabilities are equally assigned to the top-$r$-proportion solutions and remain zero for others. Consider Theorem~\ref{the:1}, which indicates the maximum probability that can be assigned by GM-QAOA, we derive the upper bound on the approximation ratio that the GM-QAOA can achieve, as stated in the following theorem:
\begin{theorem}
    \label{the:3}
    Given a problem $(C, F)$, where the top-$r$-proportion mean-max ratio is considered as $\mu_r$.
    Then, a depth-$p$ GM-QAOA circuit can achieve an approximation ratio, $\alpha$, with an upper bound given by,
    \begin{equation}
        \alpha \leqslant \mu_{\frac{1}{(2p+1)^2}}.
    \end{equation}
\end{theorem}
\begin{proof}
    See Appendix~\ref{sec:profThe3}.
\end{proof}

In contrast to the optimality density, characterizing the top-$r$-proportion mean-max ratio $\mu_r$ for the distribution is not straightforward. Through regression analysis on $\mu_r$, we establish the relationship between problem size, circuit depth, and approximation ratio $\alpha$, thus evaluating the scalability for the GM-QAOA from the aspect of $\alpha$, as discussed in Section~\ref{sec:alpha}.


\begin{figure}[t]
    \centering
    \begin{subfigure}[b]{0.32\textwidth}
        \centering
        \includegraphics[width=\textwidth]{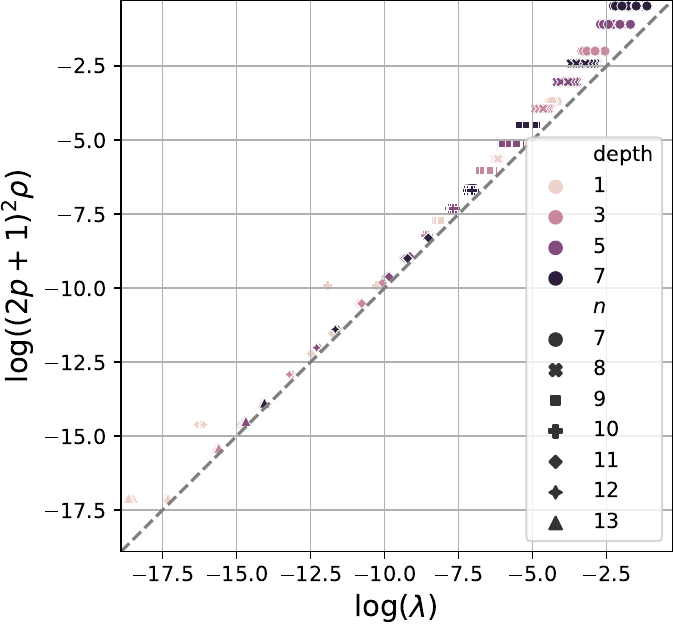} 
        \caption{travelling salesman problem}
        \label{fig:popt_sub1}
    \end{subfigure}
    \begin{subfigure}[b]{0.32\textwidth}
        \centering
        \includegraphics[width=\textwidth]{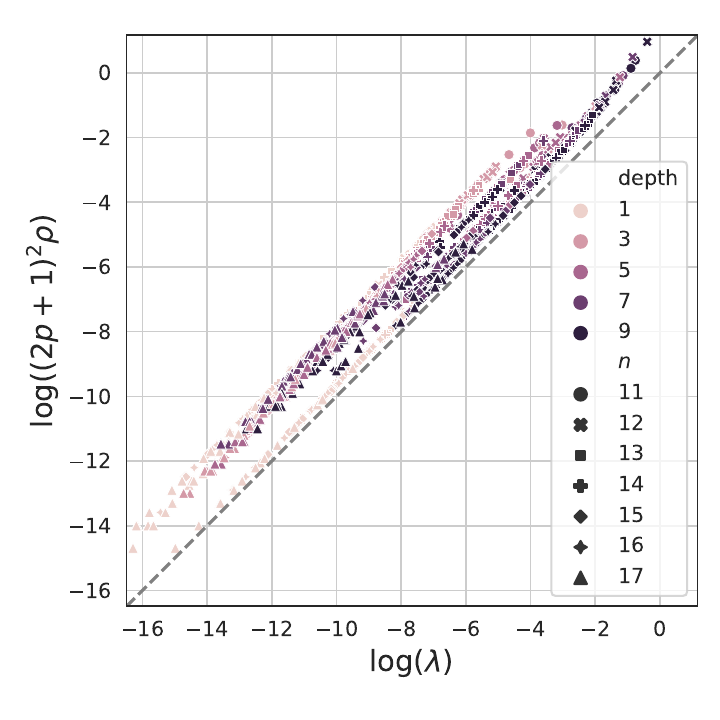}
        \caption{max-$k$-colorable-subgraph}
        \label{fig:popt_sub2}
    \end{subfigure}
    \begin{subfigure}[b]{0.32\textwidth}
        \centering
        \includegraphics[width=\textwidth]{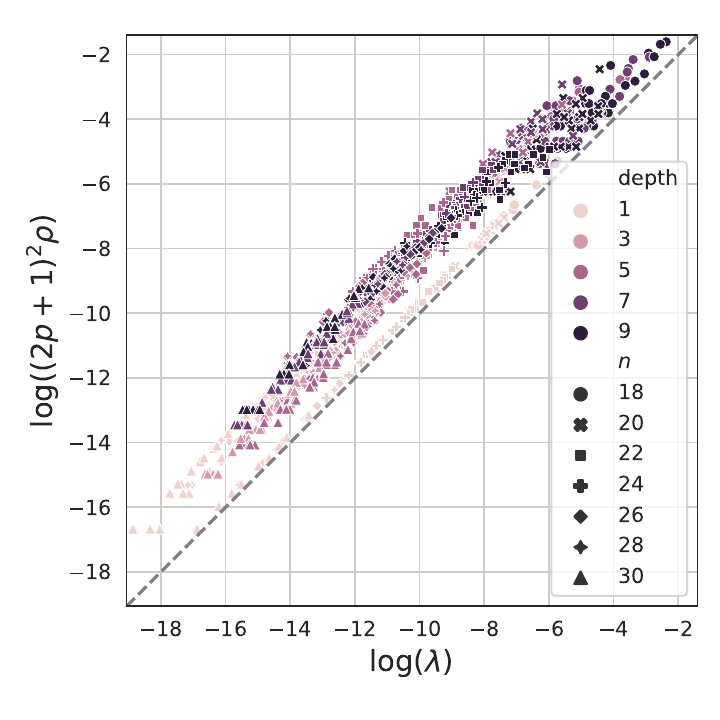}
        \caption{max-$k$-vertex-cover}
        \label{fig:popt_sub3}
    \end{subfigure}
    \caption{Comparison of the probability of sampling the optimal solution from the GM-QAOA circuit against the upper bound derived in Theorem~\ref{the:2}, where $\rho$ denotes the optimality density of the objective value distribution. The dashed gray line represents the line of equality. Here, the circuit parameters are tuned by maximizing $\lambda$, directly. All data points are positioned in the upper left and near the line of equality, illustrating the accuracy and tightness of the derived bounds.} 
    \label{fig:popt}
\end{figure}

\subsection{Optimal Solution Sampling Probability Scaling}
\label{sec:lambda}
In this section, we first validate the derived upper bound in Theorem~\ref{the:2}, for the probability of measuring the optimal solution $\lambda$, by comparing it with the $\lambda$ sampled from optimized GM-QAOA at depths ranging from $1$ to $9$. The experiments are conducted across three specific problems, where the search space can be limited to feasible sets when using GM-QAOA: traveling salesman problem, max-$k$-colorable-subgraph, and max-$k$-vertex-cover. The $k$ values of the max-$k$-colorable-subgraph and max-$k$-vertex-cover problems are set to $3$ and $\frac{n}{2}$, respectively. Here, $n$ is the graph size. The results depicted in Figure~1 show that all sampled $\lambda$ values are bounded by the theoretical upper bound for all three problems, confirming our Theorem~\ref{the:2}. Furthermore, the data points are generally close to the line of equality, demonstrating the tightness. 

\begin{figure}[t]
    \centering
    \begin{subfigure}[b]{0.36\textwidth}
        \centering
        \includegraphics[width=\textwidth]{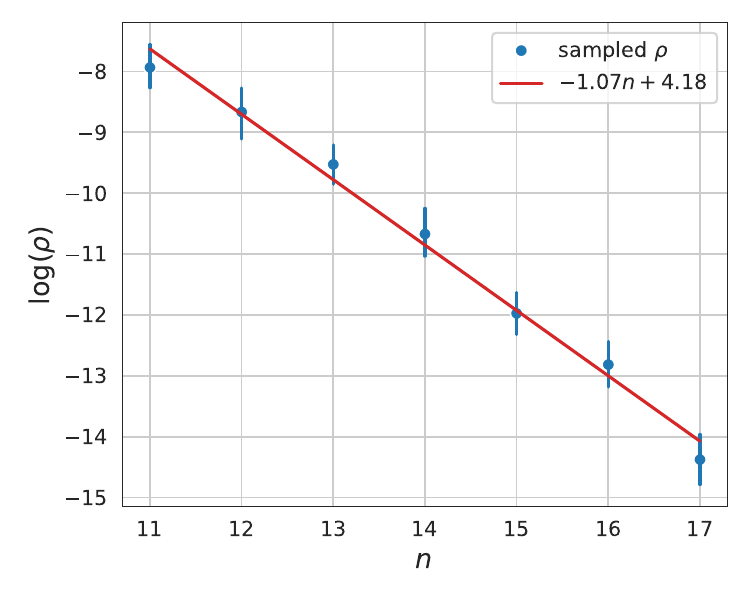} 
        \caption{max-$k$-colorable-subgraph}
        \label{fig:opt_density_sub1}
    \end{subfigure}
    \begin{subfigure}[b]{0.36\textwidth}
        \centering
        \includegraphics[width=\textwidth]{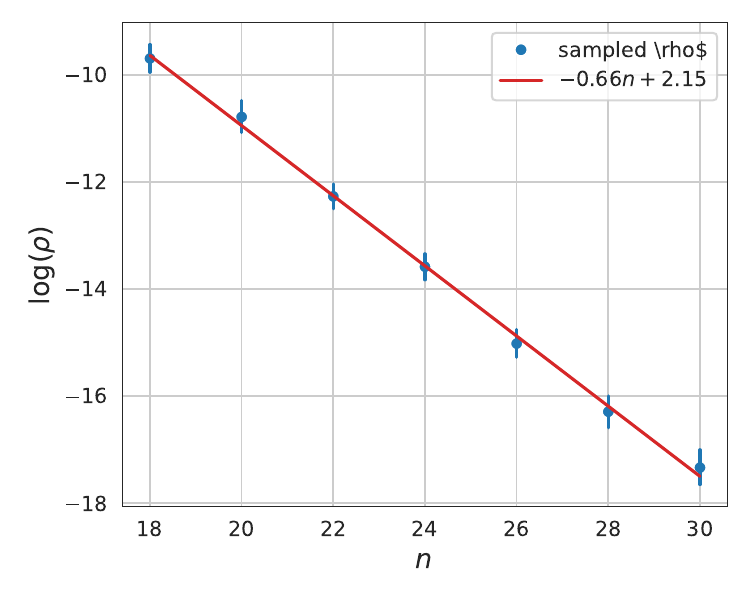}
        \caption{max-$k$-vertex-cover}
        \label{fig:opt_density_sub2}
    \end{subfigure}
    \caption{Optimality density $\rho$ sampling from max-$k$-colorable problems (\subref{fig:opt_density_sub1}) and max-$k$-vertex problems (\subref{fig:opt_density_sub2}), across various problem sizes (i.e., number of graph vertices, $n$), with $48$ instances for each size. Blue points represent the sampled data, with error bars indicating the $0.95$ confidence interval for the mean. Red lines show the linear regression results, illustrating an exponential decrease in optimality density.}
    \label{fig:opt_density}
\end{figure}

\begin{figure}[t]
    \centering
    \begin{subfigure}[b]{0.32\textwidth}
        \centering
        \includegraphics[width=\textwidth]{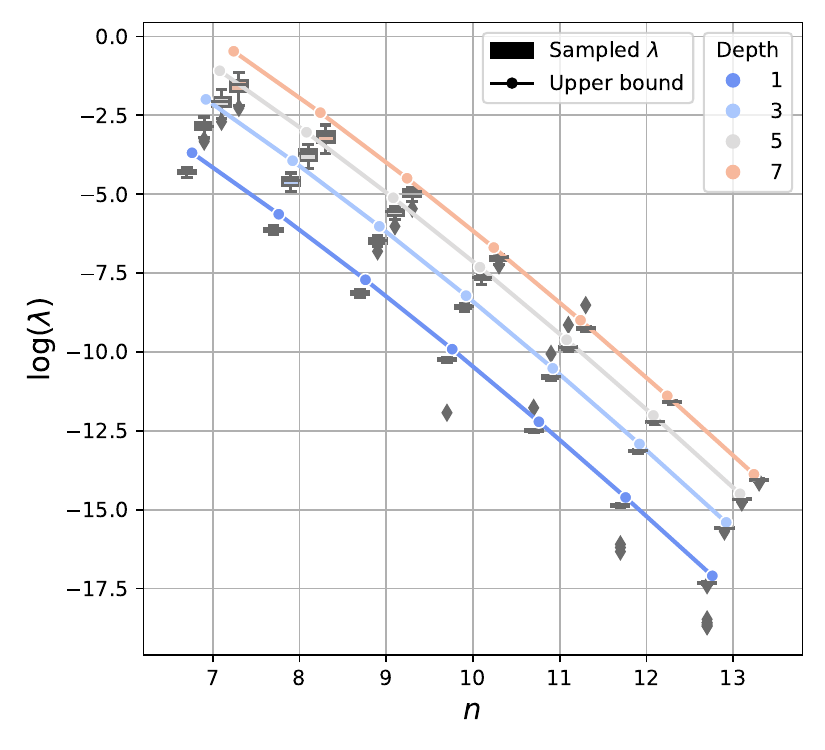} 
        \caption{travelling salesman problem}
        \label{fig:popt_fit_sub1}
    \end{subfigure}
    \begin{subfigure}[b]{0.32\textwidth}
        \centering
        \includegraphics[width=\textwidth]{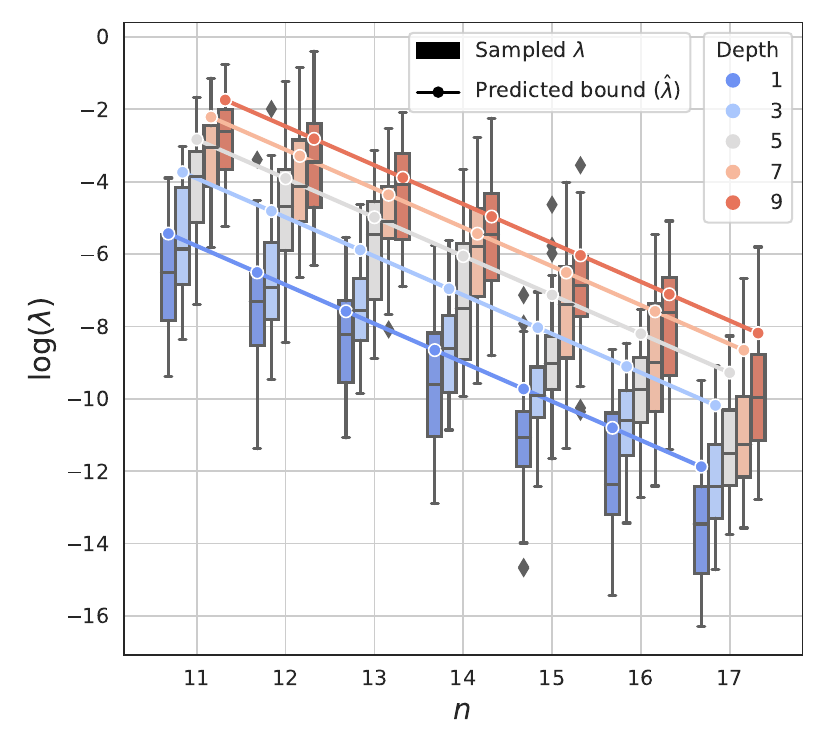}
        \caption{max-$k$-colorable-subgraph}
        \label{fig:popt_fit_sub2}
    \end{subfigure}
    \begin{subfigure}[b]{0.32\textwidth}
        \centering
        \includegraphics[width=\textwidth]{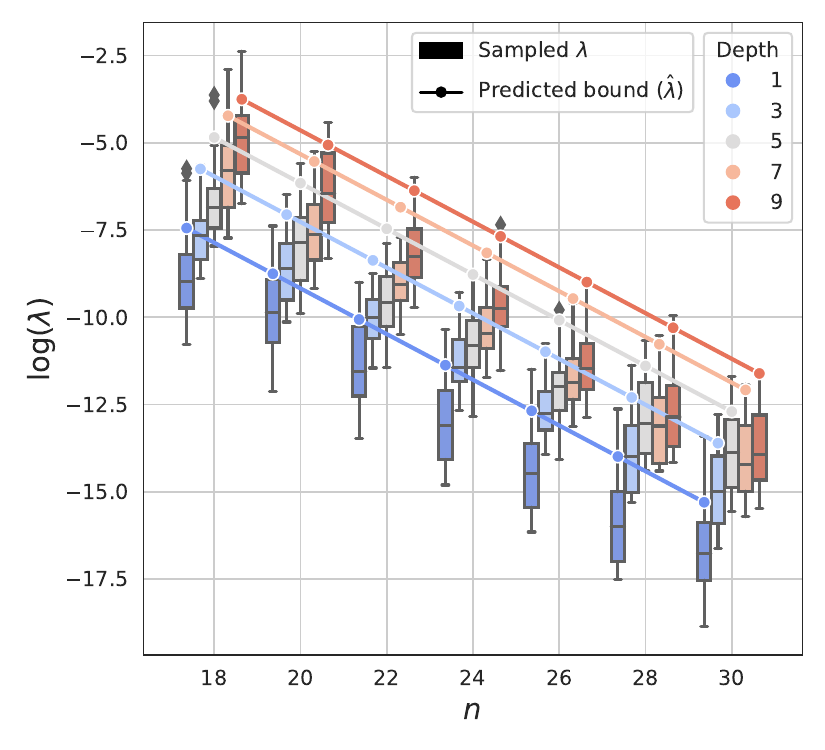}
        \caption{max-$k$-vertex-cover}
        \label{fig:popt_fit_sub3}
    \end{subfigure}
    \caption{Comparison of the probability of sampling the optimal solution $\lambda$ sampled from GM-QAOA circuits with their predictive upper bounds, based on problem size $n$ and circuit depth $p$. Upper bounds for the traveling salesman problems are calculated using Equation~(\ref{eqn:tspub}), while those for other problem types are derived from Equation~(\ref{eqn:poptmodel}). Both line plots of predictive upper bounds and the $\lambda$ values in box plots demonstrate a similar decreasing trend. This consistent trend confirms the reliability of the predictive models and further validates the scalability issues of the GM-QAOA.} 
    \label{fig:popt_pred}
\end{figure}

With the help of our tight upper bound, we then establish the relationship between problem size $n$, circuit depth $p$, and $\lambda$ upper bound, by examining the trend in optimality density $\rho$ as $n$ increases. The solutions of the traveling salesman problems are encoded in permutation matrices, using $(n-1)^2$ bits, where $n$ represents the number of locations \cite{lucas2014ising}. When the distances between locations vary, there are only $2$ optimal solutions. Thus, the optimality density is given by,  
\begin{equation}
    \rho = \frac{2}{(n-1)!},
\end{equation}
and for a depth-$p$ GM-QAOA circuit, the upper bound on the probability of sampling the optimal solution is
\begin{equation}
    \lambda < \frac{2(2p+1)^2}{(n-1)!}.
    \label{eqn:tspub}
\end{equation}
Unlike the traveling salesman problem, for the max-$k$-colorable-subgraph and max-$k$-vertex-cover problems, we cannot directly derive optimality density $\rho$ from $k$ and $n$, instead, we numerically examine $\rho$ against $n$, where $n$ is the graph vertices number. 
As illustrated in Figure~\ref{fig:opt_density}, the optimality density $\rho$ for both problems shows an exponential decrease as the problem size $n$ increases. Using linear regression on $\log(\rho)$ against $n$, we develop a predictive model for the upper bound of $\lambda$ as,
\begin{align}
    \hat{\lambda}_{\boldsymbol{\theta}}(n,p) \coloneqq \min\left((2p+1)^2 e^{\theta_1 n + \theta_2},1\right)
    \label{eqn:poptmodel}
\end{align}
where $\boldsymbol{\theta} = [\theta_1, \theta_2]$ is the regression coefficient. From Equation (\ref{eqn:tspub}) and (\ref{eqn:poptmodel}), increasing circuit depth $p$ offers a quadratic enhancement on the $\lambda$. However, this enhancement cannot offset the rapid decrease caused by the factorial term $\frac{1}{(n-1)!}$ (in (\ref{eqn:tspub})) or exponential term $e^{\theta_1 n}$ (in (\ref{eqn:poptmodel})) with a negative $\theta_1$, as the problem size $n$ grows larger, leading to scalability issues and potentially high computational costs for large problem instances.

Figure~\ref{fig:popt_pred} presents the comparison between the predictive upper bounds and $\lambda$ values sampled from optimized GM-QAOA circuits. The predictive upper bounds exhibit a consistent downward trend relative to the sampled $\lambda$ values, validating the predictive model that maps $n$ and $p$ to the $\lambda$ upper bound. Meanwhile, as the problem size increases, the benefits of increasing the circuit depth progressively diminish on a logarithmic scale, further demonstrating the exponential growth in resource requirements for GM-QAOA to maintain a constant $\lambda$.

\begin{figure}[t]
    \centering
    \begin{subfigure}[b]{0.32\textwidth}
        \centering
        \includegraphics[width=\textwidth]{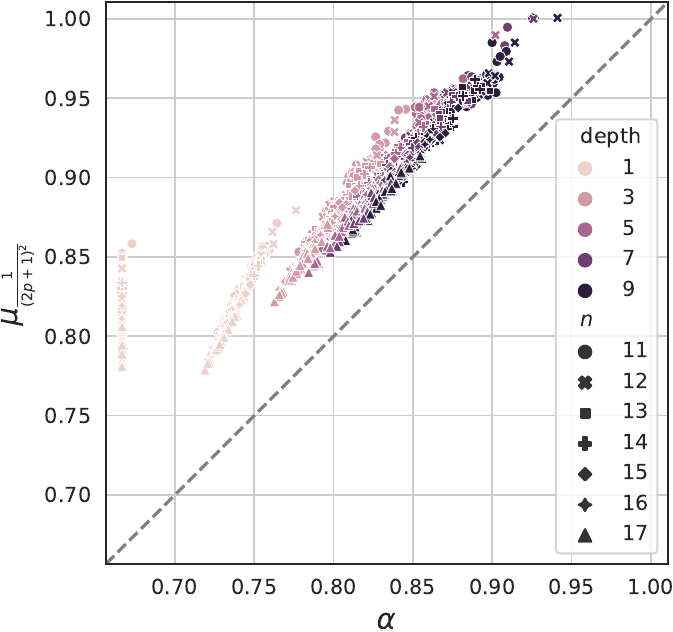} 
        \caption{max-$k$-colorable-subgraph}
        \label{fig:approx_sub1}
    \end{subfigure}
    \begin{subfigure}[b]{0.32\textwidth}
        \centering
        \includegraphics[width=\textwidth]{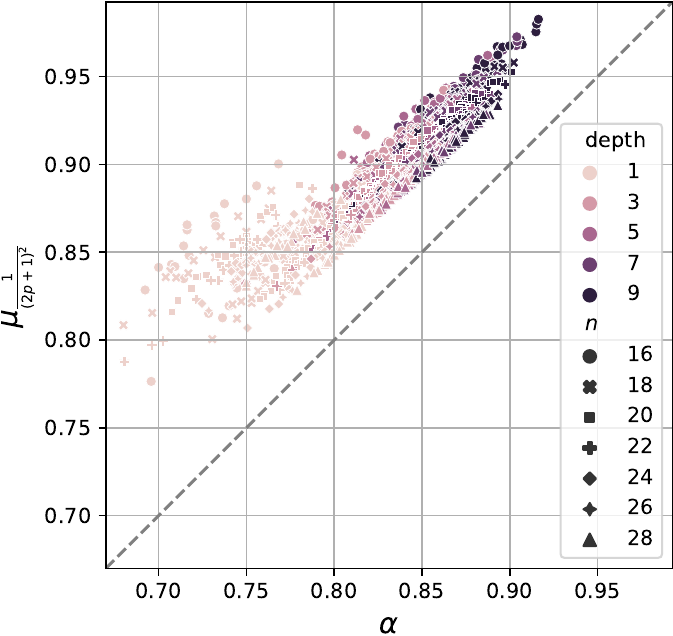}
        \caption{max-cut}
        \label{fig:approx_sub2}
    \end{subfigure}
    \begin{subfigure}[b]{0.32\textwidth}
        \centering
        \includegraphics[width=\textwidth]{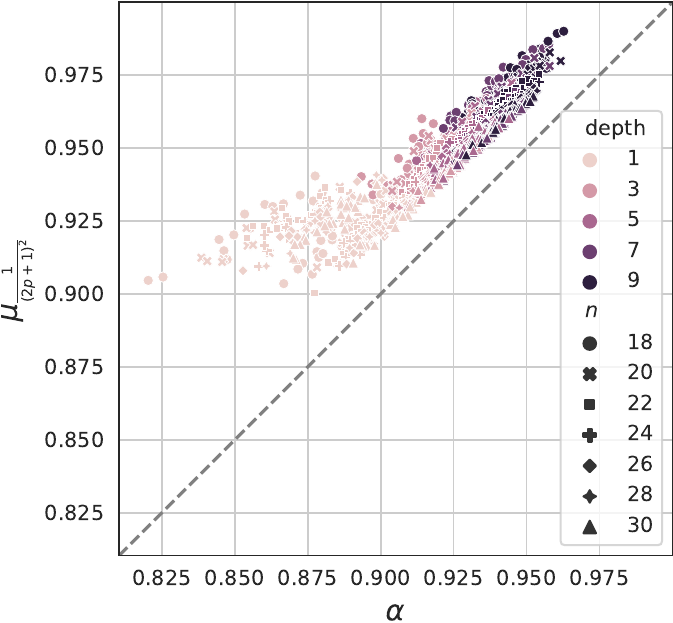}
        \caption{max-$k$-vertex-cover}
        \label{fig:approx_sub3}
    \end{subfigure}
    \caption{Comparison of approximation ratio $\alpha$ of the GM-QAOA circuit against the upper bound derived in Theorem~\ref{the:3}, where $\mu_r$ denotes the top-$r$-proportion mean-max ratio. Here, $r$ is set to $\frac{1}{(2p+1)^2}$, where p is the circuit depth. The dashed gray line represents the line of equality. From the scattering of data points, our theoretical bounds are consistent with empirical results.} 
    \label{fig:approx}
\end{figure}

\subsection{Approximation Ratio Scaling}
\label{sec:alpha}
This section expands our investigation into the approximation ratio, $\alpha$. In Figure~\ref{fig:approx}, we examine $\alpha$ values obtained from optimized GM-QAOA circuits for three maximizing optimization problems: max-$k$-colorable-subgraph, max-cut, and max-$k$-vertex-cover. Each point on the plot pairs sampled $\alpha$ values with their theoretical upper bounds, as derived in Theorem~\ref{the:3}. The data points uniformly cluster upper left of the line of equality, maintaining similar margins. This consistent positioning validates both the accuracy and robustness of the upper bound.

Having validated the upper bound, we next explore the relationship between the problem size $n$, circuit depth $p$, and $\alpha$ upper bound, by identifying the pattern of the top-$r$-proportion mean-max ratio $\mu_r$, which the upper bound depends on. Figure~\ref{fig:mu} visualizes how $\mu_r$ varies with the changes in the problem size $n$, and $\log(\frac{1}{r})$, where data points with similar $\mu_r$ values are connected by lines. The plot shows that for a constant $\mu_r$, $\log(\frac{1}{r})$ exhibits an approximately linear behavior with respect to problem size $n$. We further assume that $\log(\frac{1}{r})$ has a quadratic relationship with $\mu_r$, then the fitting model for $\mu_r$ can be built as: 
\begin{align}
    \hat{\mu}_{\boldsymbol{\theta}}(n,r) \coloneqq \sqrt{\frac{-\log(r)}{\theta_1 n + \theta_2}} + \frac{\theta_3}{1+e^{-\theta_4(n-\theta_5)}},
    \label{eqn:approxmodel1}
\end{align}
where $\boldsymbol{\theta} = [\theta_1, \theta_2, \dots, \theta_5]$ is the regression parameter.
For the max-$k$-colorable-subgraph problems, since the sampled distribution exhibits the same value for $\mu_1$, we adjust the model as: 
\begin{align}
    \hat{\mu}_{\boldsymbol{\theta}}(n,r) \coloneqq \sqrt{\frac{-\log(r)}{\theta_1 n + \theta_2}} + \theta_3,
    \label{eqn:approxmodel2}
\end{align}
where $\theta_3$ is specifically set to equal $\mu_1$.
The fitting results are presented in Table~\ref{tab:reg} and visualized in Figure~\ref{fig:mu_fit}.

\begin{figure}[t]
    \centering
    \begin{subfigure}[b]{0.32\textwidth}
        \centering
        \includegraphics[width=\textwidth]{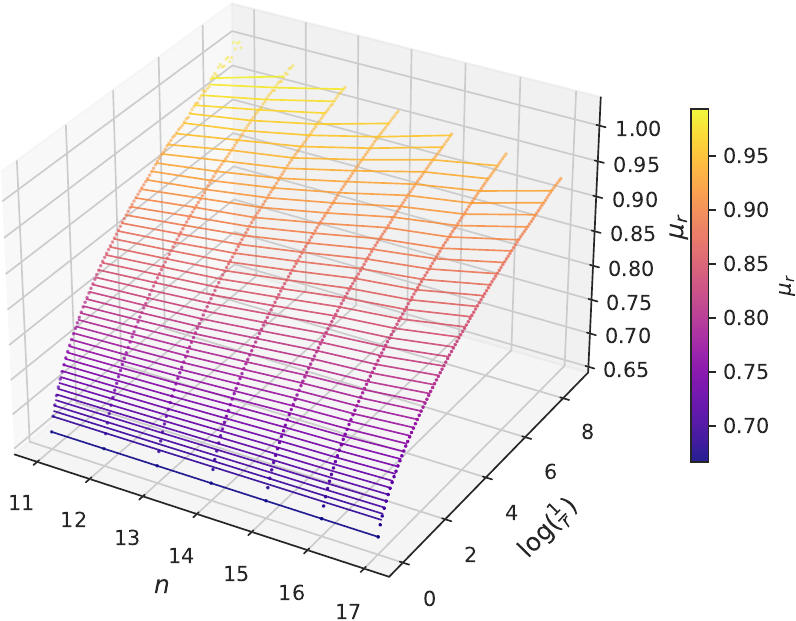} 
        \caption{max-$k$-colorable-subgraph}
        \label{fig:mu_sub1}
    \end{subfigure}
    \begin{subfigure}[b]{0.32\textwidth}
        \centering
        \includegraphics[width=\textwidth]{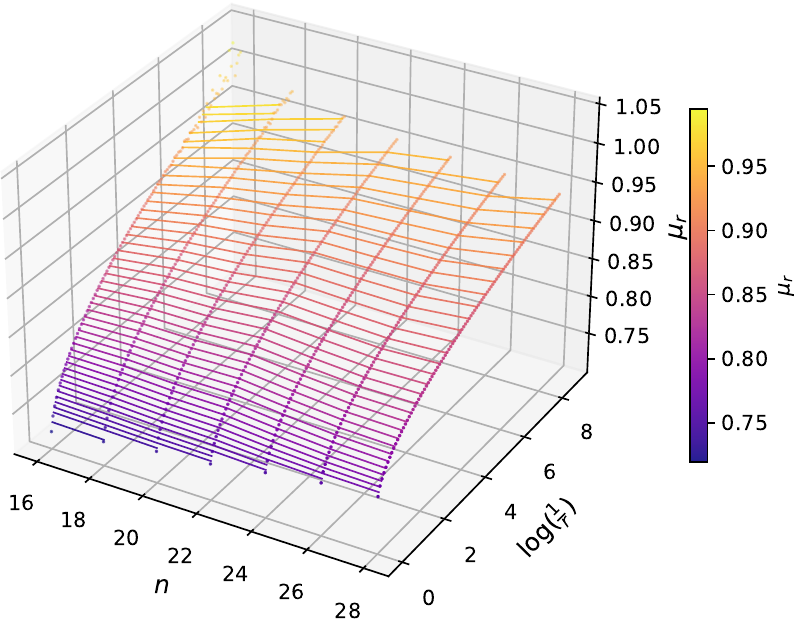}
        \caption{max-cut}
        \label{fig:mu_sub2}
    \end{subfigure}
    \begin{subfigure}[b]{0.32\textwidth}
        \centering
        \includegraphics[width=\textwidth]{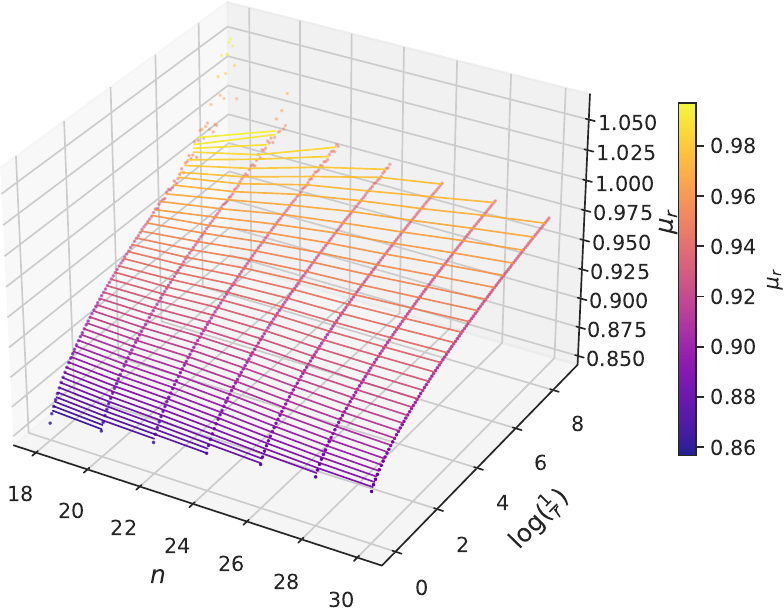}
        \caption{max-$k$-vertex-cover}
        \label{fig:mu_sub3}
    \end{subfigure}
    \caption{Top-$r$-proportion mean-max ratio $\mu_r$ of max-$k$-colorable-subgraph, max-cut and max-$k$-vertex-cover problem. 
    Each data point is composed of $(n, \log(\frac{1}{r}), \overline{\mu}_r)$, where $n$ is the problem size, 
    and $\overline{\mu}_r$ is the averaged value of $\mu_r$ over 48 different instances of the respective problems. 
    The lines connect data points with the similar $\overline{\mu}_r$ value. The color bar maps the colors of lines and data points to the corresponding $\overline{\mu}_r$ value.} 
    \label{fig:mu}
\end{figure}

\begin{figure}[t]
    \centering
    \begin{subfigure}[b]{0.32\textwidth}
        \centering
        \includegraphics[width=\textwidth]{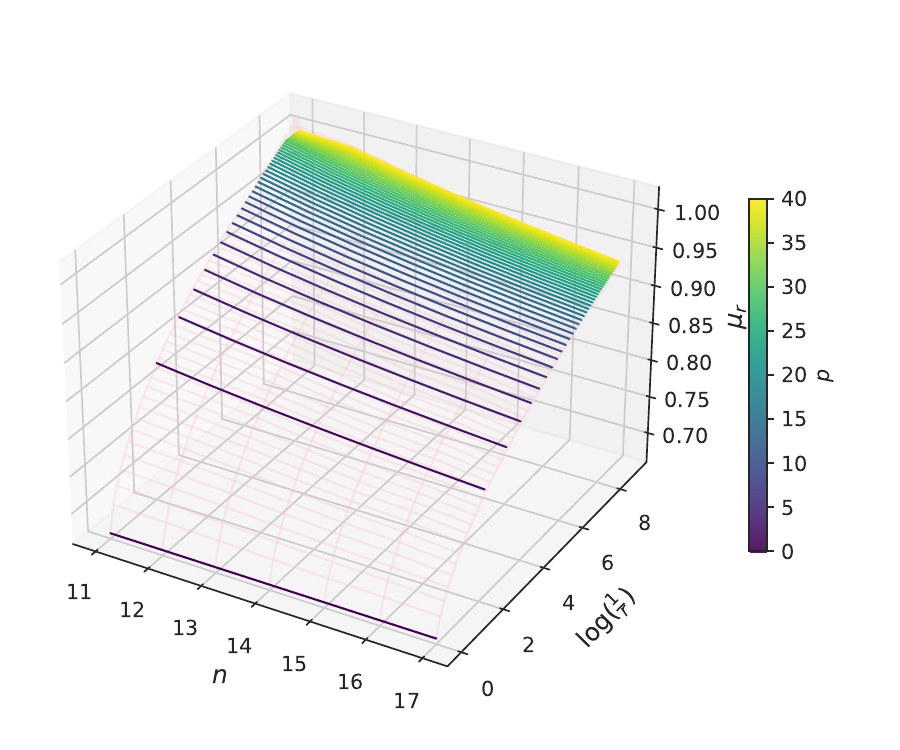} 
        \caption{max-$k$-colorable-subgraph}
        \label{fig:mu_fit_sub1}
    \end{subfigure}
    \begin{subfigure}[b]{0.32\textwidth}
        \centering
        \includegraphics[width=\textwidth]{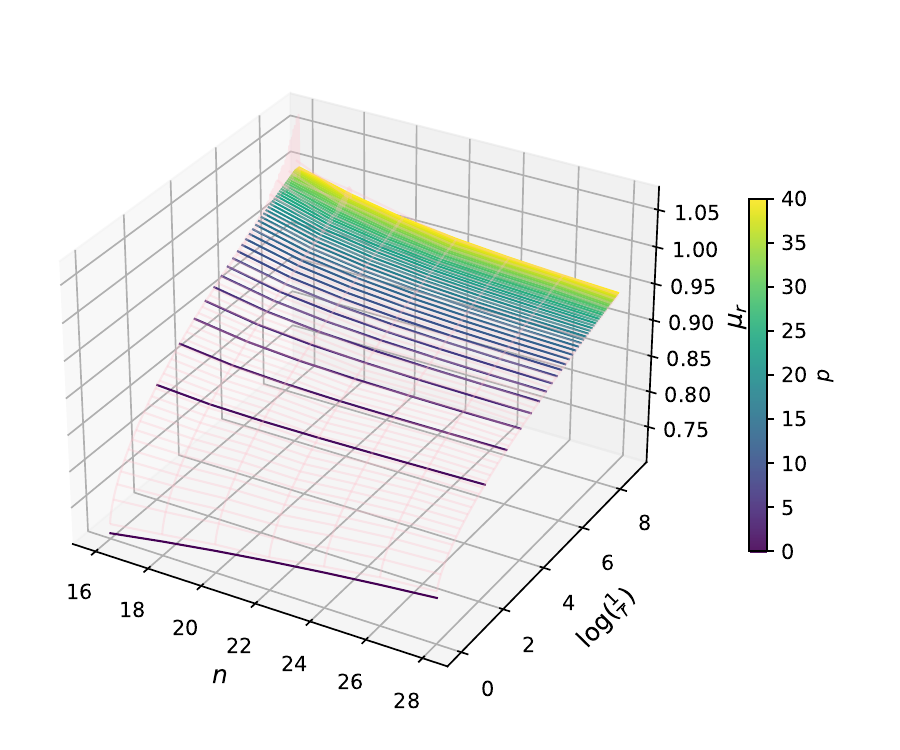}
        \caption{max-cut}
        \label{fig:mu_fit_sub2}
    \end{subfigure}
    \begin{subfigure}[b]{0.32\textwidth}
        \centering
        \includegraphics[width=\textwidth]{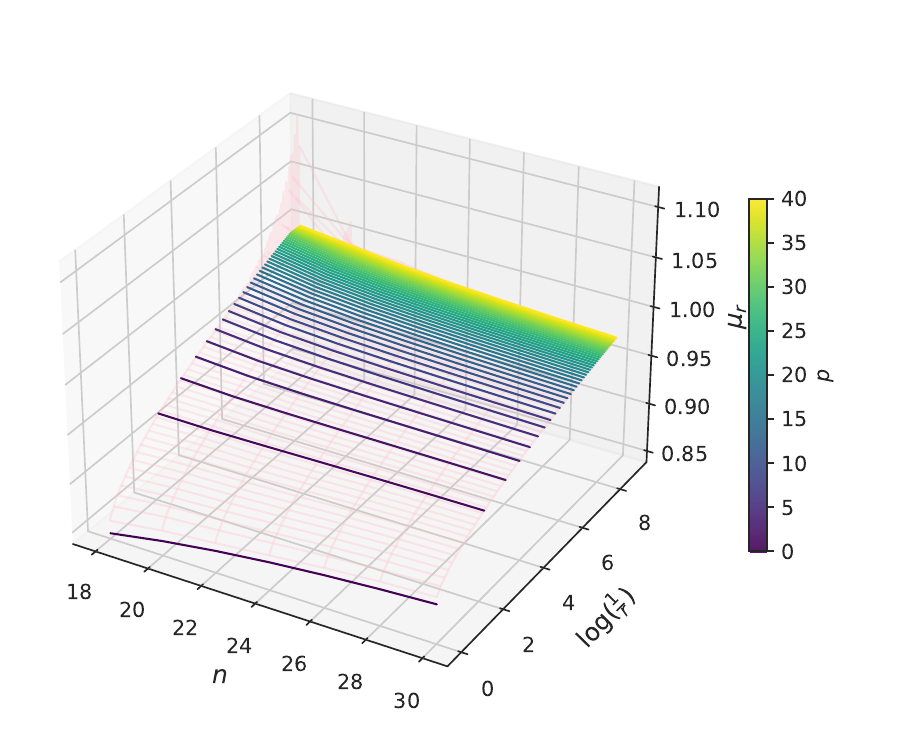}
        \caption{max-$k$-vertex-cover}
        \label{fig:mu_fit_sub3}
    \end{subfigure}
    \caption{Regression results for the top-$r$-proportion mean-max ratio $\mu_r$. 
    The averaged values of sampled $\mu_r$ are visualized using pink wireframes. 
    The fitted values, $\hat{\mu}_{\boldsymbol{\theta}}(n,r)$, where $r$ is set to $\frac{1}{(2p+1)^2}$, 
    are depicted using lines colored with the Viridis palette. 
    The color bar maps these colors to their corresponding $p$ values.} 
    \label{fig:mu_fit}
\end{figure}

Using the fitted model for $\mu_r$, we can now predict the upper bound for the approximation ratio achieved by depth-$p$ GM-QAOA as,
\begin{align}
    \hat{\alpha}_{\boldsymbol{\theta}}(n,p) \coloneqq \min\left(\hat{\mu}_{\boldsymbol{\theta}}\left(n,\frac{1}{(2p+1)^2}\right),1\right).
\end{align}
Figure~\ref{fig:fit_approx} presents a comparative analysis between $\hat{\alpha}$ and the empirically obtained approximation ratios $\alpha$ from optimized GM-QAOA circuits across depths ranging from $1$ to $9$. 
The results show that most of the predicted upper bounds are higher than the true values of $\alpha$, making these upper bound predictions reliable, and validating the effectiveness of the predictive model. From the definitions of the fitting models (\ref{eqn:approxmodel1}) and (\ref{eqn:approxmodel2}), the second term converges to or is originally a constant. Since the enhancement achieved by increasing $p$ is trapped in a logarithm, it becomes challenging to counterbalance the decrease in $\alpha$ caused by the scaling up of the problem size. Consequently, maintaining a consistent approximation ratio as the problem size increases would require the depth of GM-QAOA to grow exponentially. Based on this observation, we propose the following conjecture.
\begin{conjecture}
    Given a family of problems characterized by a certain objective function structure. As the problem size increases, achieving a target approximation ratio greater than a certain value requires the depth of GM-QAOA to grow exponentially with respect to the problem size.
\end{conjecture}


\begin{table}[t]
    \centering
    \caption{Regression results for top-$r$-proportion mean-max ratio $\mu_r$}
    \begin{tabular}{@{}l
                    S[table-format=1.2e+1]
                    S[table-format=2.2e+1]
                    S[table-format=1.2e-1]
                    S[table-format=1.2e-1]
                    S[table-format=2.2e+1]@{}}
    \toprule
    & \multicolumn{5}{c}{parameters of regression model} \\
    \cmidrule(lr){2-6}
    & {$\theta_1$} & {$\theta_2$} & {$\theta_3$} & {$\theta_4$} & {$\theta_5$} \\
    \midrule
    max-$k$-colorable-subgraph & 7.68e0  & -1.12e1 & 6.67e-1  & {}  & {}  \\
    max-cut                    & 1.38e1  & -1.20e2 & 8.09e-1  & 7.19e-2  & -1.10e1 \\
    max-$k$-vertex-cover       & 5.21e1  & -5.98e2 & 8.90e-1  & 1.20e-1  & -6.21e0 \\
    \bottomrule
    \end{tabular}
    \label{tab:reg}
\end{table}

\begin{figure}[t]
    \centering
    \begin{subfigure}[b]{0.32\textwidth}
        \centering
        \includegraphics[width=\textwidth]{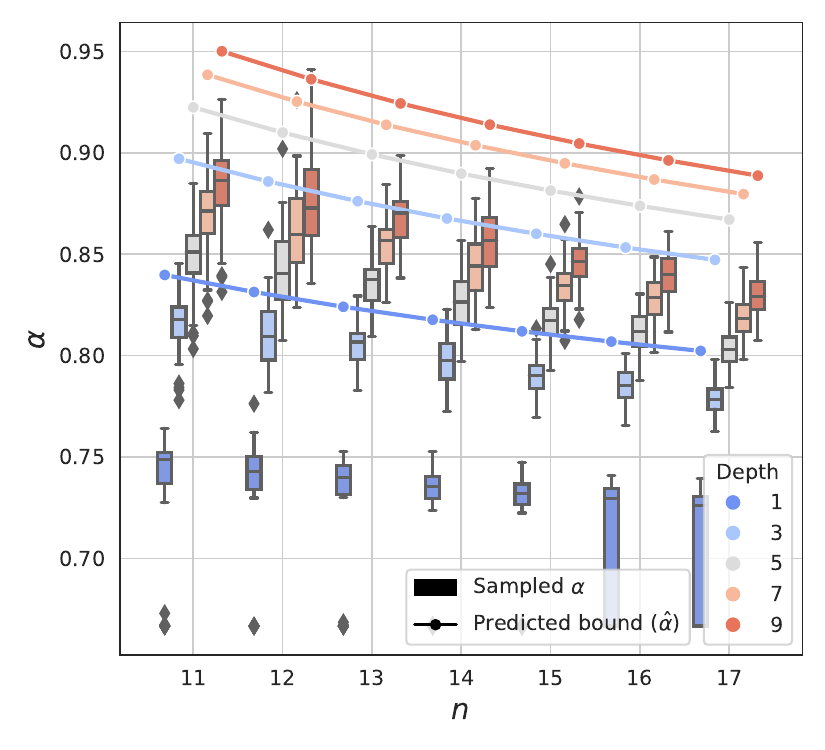} 
        \caption{max-$k$-colorable-subgraph}
        \label{fig:fit_approx_sub1}
    \end{subfigure}
    \begin{subfigure}[b]{0.32\textwidth}
        \centering
        \includegraphics[width=\textwidth]{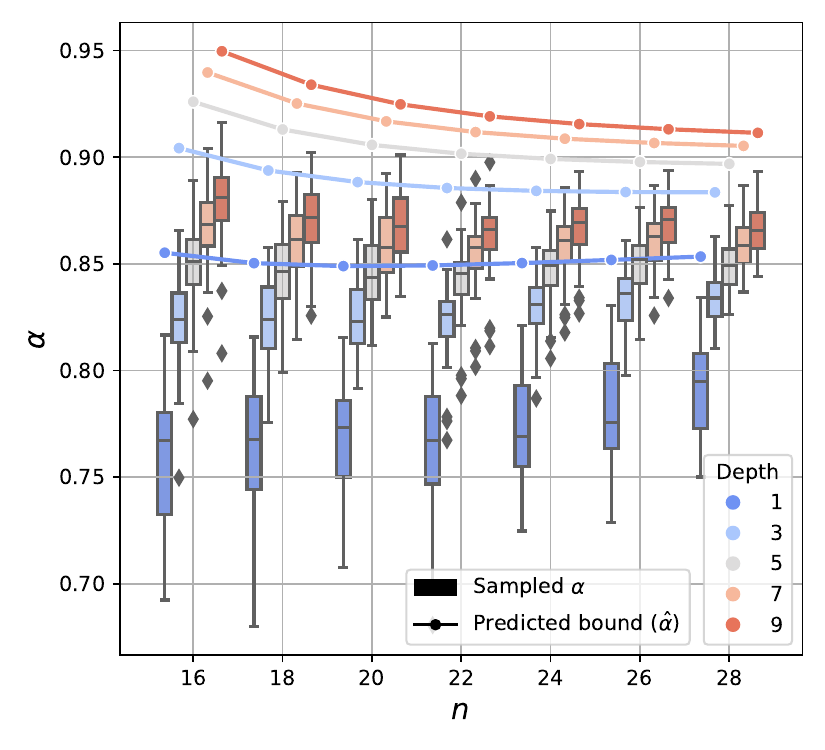}
        \caption{max-cut}
        \label{fig:fit_approx_sub2}
    \end{subfigure}
    \begin{subfigure}[b]{0.32\textwidth}
        \centering
        \includegraphics[width=\textwidth]{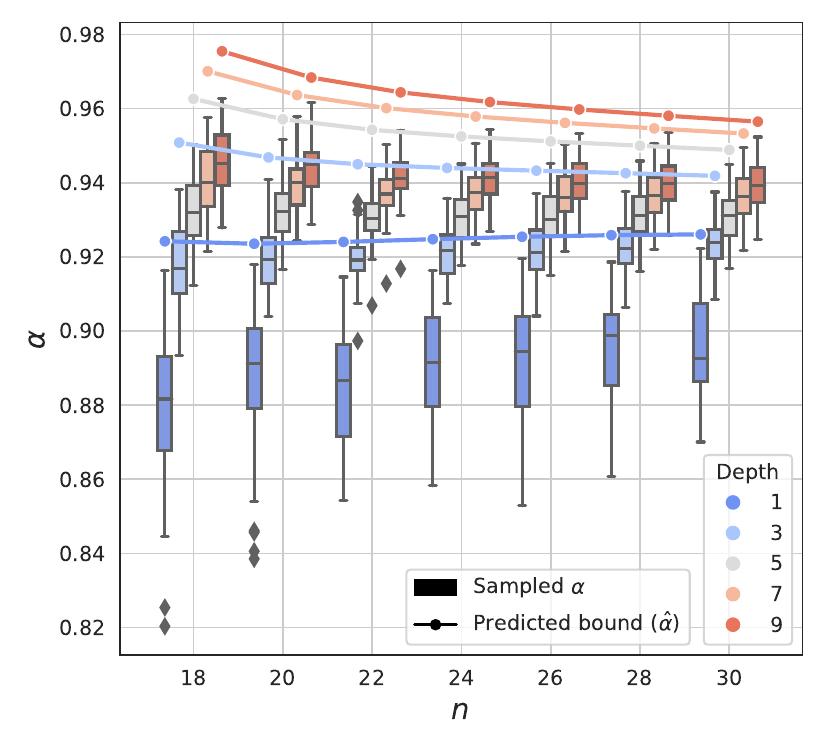}
        \caption{max-$k$-vertex-cover}
        \label{fig:fit_approx_sub3}
    \end{subfigure}
    \caption{Comparison of the approximation ratios $\alpha$ sampled from GM-QAOA circuits against the predictive upper bounds defined in Equations (\ref{eqn:approxmodel1}) and (\ref{eqn:approxmodel2}). The box plots representing sampled $\alpha$ values exhibit a trend consistent with the predictive line plots, with the sampled values nearly always falling within the bounds established by the line plots. This consistent tight bounding supports our conjecture that GM-QAOA also faces scalability challenges for maintaining approximation ratios as problem sizes increase.} 
    \label{fig:fit_approx}
\end{figure}

\section{Conclusion}
In this work, we prove an upper bound $\frac{(2p+1)^2}{\vert F \vert}$ on the probability of measuring a computational basis state from a depth-$p$ GM-QAOA state, where $\vert F \vert$ denotes the number of solutions within the search space. Based on this result, we derive upper bounds on two performance metrics for the GM-QAOA: the probability of sampling the optimal solution $\lambda$, and the approximation ratio $\alpha$. These upper bounds are formulated in terms of the statistical metrics we defined for objective value distribution, namely the optimality density $\rho$ and the top-$r$-proportion mean-max ratio $\mu_r$, respectively.

We validate the theoretical performance bounds by comparing them against $\lambda$ and $\alpha$ values obtained from optimized GM-QAOA circuits, confirming their reliability and tightness. Following this, our regression analysis on $\rho$ and $\mu_r$ further establishes the relationships between the problem size, circuit depth, and the performance limits of $\lambda$ and $\alpha$. The results indicate scalability challenges for the GM-QAOA. As the enhancement in measurement probability offered by increasing the depth of GM-QAOA is at best quadratic, GM-QAOA would necessitate an exponential increase in circuit resources to maintain a consistent level of performance in either the aspect of $\lambda$ and $\alpha$.

\begin{appendices}

\section{Proof of Theorem 1}
\label{sec:profThe1}
In this section, 
we present the derivation of our main result, Theorem~\ref{the:1}, 
which provides the upper bound on the probability of measuring a computational basis state from a GM-QAOA state with a given depth. 
We begin by introducing two key notations and deriving two essential lemmas.

\begin{notation}
    Define $\binom{[p]}{k}$ as the set of all possible combinations of choosing $k$ elements from $[p] \coloneqq \{1, 2, \cdots, p\}$,
    where each combination $\mathbf{s} \in \binom{[p]}{k}$ is a vector whose entries are in ascending order.
\end{notation}
\begin{notation}
    Let $S$ be a set of vectors, where each $\mathbf{s} = [s_1,s_2,\cdots,s_{\vert \mathbf{s} \vert}]$. 
    Define the operations $\lhd_nS$ and $S\rhd_n$ as follows, 
    \begin{equation}
        \begin{split}
            & \lhd_nS \coloneqq \left\{ [n, s_1,s_2,\cdots,s_{\vert \mathbf{s} \vert}] \middle\vert \mathbf{s} \in S \right\},\; (s_0 \coloneqq n, \forall \mathbf{s} \in \lhd_n S),\\
            & S\rhd_n \coloneqq \left\{ [s_1,s_2,\cdots,s_{\vert \mathbf{s}\vert}, n ] \middle\vert \mathbf{s} \in S \right\},
        \end{split}
    \end{equation}
    where $\lhd_nS$ prepends and $S\rhd_n$ appends the number $n$ to every vector in the set $S$, respectively, resulting in new sets of vectors.
\end{notation}

\begin{lemma}
    \label{lemma:umcancel}
    Consider a depth-$p$ GM-QAOA circuit with parameters $\boldsymbol{\beta} \in S_m, S_m \coloneqq \bigcup_{\mathbf{s}\in \binom{[p]}{m}}\{\boldsymbol{\beta} \mid \beta_i = 2k_i \pi, k_i \in \mathbb{Z}, \forall i \in \mathbf{s}; \beta_j \in \mathbb{R}, \forall j \notin \mathbf{s}\}$, 
    containing $m$ elements as integral multiples of $2\pi$, while the remaining $p-m$ parameters are any real number, 
    then, this depth-$p$ circuit reduces to a depth-$(p\!-\!m)$ circuit.
\end{lemma}
\begin{proof}
    From Equation~\ref{eqn:UM}, for any search space $F$, operation $U^{(M)}_F(2k\pi) = I$, $\forall k \in \mathbb{Z}$, 
    resulting in no alteration to the circuit's behavior.

    When $\beta_p = 2k\pi, k \in \mathbb{Z}$, the circuit, 
    \begin{equation}
        \vert \psi_{p,\mathcal{C},F}(\boldsymbol{\gamma},\boldsymbol{\beta})\rangle = U^{(P)}_{\mathcal{C}}(\gamma_p) U^{(M)}_F(\beta_{p-1}) U^{(P)}_{\mathcal{C}}(\gamma_{p-1}) \cdots U^{(M)}_F(\beta_1)U^{(P)}_{\mathcal{C}}(\gamma_1)\vert F \rangle.
    \end{equation}
    According to Equation~\ref{eqn:UP}, 
    the last operation, $U^{(P)}_{\mathcal{C}}(\gamma_p)$ merely shifts the phase of each computational basis state without altering the measurement probabilities. 
    Consequently, operation $U^{(M)}_F(\beta_p) U^{(P)}_{\mathcal{C}}(\gamma_p)$ can be considered canceled.

    When $\beta_j = 2k\pi, k \in \mathbb{Z}, j \neq p$, the circuit, 
    \begin{equation}
        \vert \psi_{p,\mathcal{C},F}(\boldsymbol{\gamma},\boldsymbol{\beta})\rangle = \cdots U^{(M)}_F(\beta_{j+1})U^{(P)}_{\mathcal{C}}(\gamma_{j+1})U^{(P)}_{\mathcal{C}}(\gamma_j)  \cdots \vert F \rangle.
    \end{equation}
    From Equation~\ref{eqn:UP}, 
    we can derive, 
    \begin{equation}
        \begin{split}
            U^{(P)}_{\mathcal{C}}(\gamma_{j+1})U^{(P)}_{\mathcal{C}}(\gamma_j)\vert f \rangle = & e^{-i(\gamma_{j+1}+\gamma_j)\mathcal{C}(f)}\vert f \rangle \\
            = & U^{(P)}_{\mathcal{C}}(\gamma_{j+1}+\gamma_j)\vert f \rangle, \quad f \in F.
        \end{split}
    \end{equation}
    Hence, let $\gamma_{j+1} \leftarrow \gamma_{j+1}+\gamma_{j}$, the operation $U^{(M)}_F(\beta_{j})U^{(P)}_{\mathcal{C}}(\gamma_{j})$ is equivalent to being canceled. 
\end{proof}

\begin{lemma}
    \label{lemma:upcancel}
    Given a problem defined with the search space $F$. Consider a depth-$p$ GM-QAOA circuit using a phase function $\mathcal{C}$, 
    with parameters $\boldsymbol{\gamma} \in S_m, S_m \coloneqq \bigcup_{\mathbf{s}\in \binom{[p]}{m}}\{\boldsymbol{\gamma} \mid \gamma_i\mathcal{C}(f) = 2k_{i,f} \pi, k_{i,f} \in \mathbb{Z}, \forall i \in \mathbf{s}, \forall f \in F; \gamma_j \in \mathbb{R}, \forall j \notin \mathbf{s}\}$, 
    then, this depth-$p$ circuit reduces to a depth-$(p\!-\!m)$ circuit.
\end{lemma}
\begin{proof}
    From Equation~\ref{eqn:UP}, if $\gamma\mathcal{C}(f) = 2k_f\pi$, $k_f \in \mathbb{Z}$, $\forall f\in F$, 
    then $U^{(P)}_{\mathcal{C}}(\gamma)\vert f \rangle  = \vert f \rangle$, $\forall f \in F$, 
    resulting in no alteration to the circuit's behavior.

    When $\gamma_1\mathcal{C}(f) = 2k_f\pi$, $k_f \in \mathbb{Z}$, $\forall f\in F$, the circuit, 
    \begin{equation}
        \begin{split}
            \vert \psi_{p,\mathcal{C},F}(\boldsymbol{\gamma},\boldsymbol{\beta})\rangle & = U^{(M)}_F(\beta_{p}) U^{(P)}_{\mathcal{C}}(\gamma_p) U^{(M)}_F(\beta_{p-1}) \cdots U^{(M)}_F(\beta_2)U^{(P)}_{\mathcal{C}}(\gamma_2) U^{(M)}_F(\beta_1) \vert F \rangle \\
            & = e^{-i\beta_1} U^{(M)}_F(\beta_{p}) U^{(P)}_{\mathcal{C}}(\gamma_p) U^{(M)}_F(\beta_{p-1}) \cdots U^{(M)}_F(\beta_2)U^{(P)}_{\mathcal{C}}(\gamma_2) \vert F \rangle.
        \end{split}
    \end{equation}
    As global phase $e^{-i\beta_1}$ is not observable when measuring, the operation $U^{(M)}_F(\beta_1)U^{(P)}_{\mathcal{C}}(\gamma_1)$ can be regarded as canceled.

    When $\gamma_j\mathcal{C}(f) = 2k_f\pi$, $k_f \in \mathbb{Z}$, $\forall f\in F$ and $j \neq 1$, the circuit, 
    \begin{equation}
        \vert \psi_{p,\mathcal{C},F}(\boldsymbol{\gamma},\boldsymbol{\beta})\rangle = \cdots U^{(M)}_F(\beta_{j})U^{(M)}_F(\beta_{j-1})U^{(P)}_{\mathcal{C}}(\gamma_{j-1})  \cdots \vert F \rangle.
    \end{equation}
    From Equation~\ref{eqn:UM}, 
    we can derive, 
    \begin{equation}
        \begin{split}
            U^{(M)}_F(\beta_{j})U^{(M)}_F(\beta_{j-1}) = & I - (1-e^{-i(\beta_j+\beta_{j-1})})\vert F \rangle \langle F \vert \\
            = & U^{(M)}_F(\beta_{j}+\beta_{j-1}).
        \end{split}
    \end{equation}
    Hence, let $\beta_{j-1} \leftarrow \beta_j+\beta_{j-1}$, the operation $U^{(M)}_F(\beta_{j})U^{(P)}_{\mathcal{C}}(\gamma_{j})$ is equivalent to being canceled. 
\end{proof}

\begin{proof}[Proof of Theorem~\ref{the:1}]
    Expand the prepared states of the GM-QAOA, we get,
    \begin{align}
            & \vert \psi_{p,\mathcal{C},F}(\boldsymbol{\gamma},\boldsymbol{\beta})\rangle \nonumber\\
            = & U^{(M)}_F(\beta_p) U^{(P)}_{\mathcal{C}}(\gamma_p) \cdots U^{(M)}_F(\beta_2) U^{(P)}_{\mathcal{C}}(\gamma_2) U^{(M)}_F(\beta_1) U^{(P)}_{\mathcal{C}}(\gamma_1)\vert F \rangle \nonumber\\
            = & \left(U^{(P)}_{\mathcal{C}}(\gamma_p) - (1-e^{-i\beta_p})\vert F \rangle \langle F \vert U^{(P)}_{\mathcal{C}}(\gamma_p)\right) \cdots \left(U^{(P)}_{\mathcal{C}}(\gamma_1) - (1-e^{-i\beta_1})\vert F \rangle \langle F \vert U^{(P)}_{\mathcal{C}}(\gamma_1)\right) \vert F \rangle \nonumber\\
            = & \sum_{\substack{\mathbf{s} \in \binom{[p]}{k},\\ k\in \{0,1,\cdots,p\}}} \left(\prod_{j=1}^{k}(e^{-i\beta_{s_j}}-1)\right) U^{(P)}_{\mathcal{C}}(\sum_{j=s_k+1}^{p}\gamma_j)\vert F \rangle\langle F \vert U^{(P)}_{\mathcal{C}}(\sum_{j=s_{k-1}+1}^{s_k}\gamma_j)\cdots \vert F \rangle\langle F \vert U^{(P)}_{\mathcal{C}}(\sum_{j=1}^{s_1}\gamma_j)\vert F \rangle \nonumber\\
            = & \sum_{\substack{\mathbf{s} \in \binom{[p]}{k},\\ k\in \{0,1,\cdots,p\}}} \left(\prod_{j=1}^{k}(e^{-i\beta_{s_j}}-1)\right) \left(U^{(P)}_{\mathcal{C}}(\sum_{j=s_k+1}^{p}\gamma_j)\vert F \rangle\right) \frac{\displaystyle\sum_{f\in F}e^{-i \sum_{j=s_{k-1}+1}^{s_k}\gamma_j \mathcal{C}(f)}}{\vert F \vert} \cdots \frac{\displaystyle\sum_{f\in F}e^{-i \sum_{j=1}^{s_1}\gamma_j \mathcal{C}(f)}}{\vert F \vert} \nonumber\\
            = & \frac{1}{\sqrt{\vert F \vert}}\sum_{f\in F} \sum_{\substack{\mathbf{s} \in \lhd_0\binom{[p]}{k},\\ k\in \{0,1,\cdots,p\}}} \left(\prod_{j=1}^{k}(e^{-i\beta_{s_j}}-1)\right) \left(\prod_{j=1}^{k}\frac{\displaystyle\sum_{f'\in F}e^{-i \sum_{l=s_{j-1}+1}^{s_j}\gamma_l \mathcal{C}(f')}}{\vert F \vert}\right) e^{-i \sum_{j=s_k+1}^{p}\gamma_j \mathcal{C}(f)} \vert f \rangle 
    \end{align}
    Then, the probability of measuring a computational basis state $\vert f\rangle$, $f \in F$, is given as,
    \begin{equation}
        \begin{split}
          & \vert \langle f \vert \psi_{p,\mathcal{C},F}(\boldsymbol{\gamma},\boldsymbol{\beta})\rangle \vert^2 \\
        = & \frac{1}{\vert F \vert} \left\vert \sum_{\substack{\mathbf{s} \in \lhd_0\binom{[p]}{k},\\ k\in \{0,1,\cdots,p\}}} e^{-i \sum_{j=s_k+1}^{p}\gamma_j \mathcal{C}(f)} \prod_{j=1}^{k}(e^{-i\beta_{s_j}}-1) \prod_{j=1}^{k}\frac{\displaystyle\sum_{f'\in F}e^{-i \sum_{l=s_{j-1}+1}^{s_j}\gamma_l \mathcal{C}(f')}}{\vert F \vert} \right\vert^2.
        \end{split}
    \end{equation}

    Here, we assume a tunable phase function $\mathcal{C}$, and introduce a vector parameter $\mathbf{v} \in \mathbb{R}^{\vert F \vert}$, 
    where each component represents $\mathcal{C}(f')$, $f' \in F$. 
    Additionally, we define another parameter $z \in \mathbb{R}$, 
    independent of $\mathbf{v}$, to specifically replace $\mathcal{C}(f)$ outside the product of sums. 
    Then, we have a new function, 
    \begin{equation}
        G_p(\boldsymbol{\gamma},\boldsymbol{\beta},\mathbf{v},z) \coloneqq \frac{1}{\vert F \vert}\vert g_p(\boldsymbol{\gamma},\boldsymbol{\beta},\mathbf{v},z) \vert^2,
    \end{equation}
    where
    \begin{equation}
        g_p(\boldsymbol{\gamma},\boldsymbol{\beta},\mathbf{v},z) \coloneqq \sum_{\substack{\mathbf{s} \in \lhd_0\binom{[p]}{k},\\ k\in \{0,1,\cdots,p\}}} e^{-i \sum_{j=s_k+1}^{p}\gamma_j z} \prod_{j=1}^{k}(e^{-i\beta_{s_j}}-1)\prod_{j=1}^{k}\frac{\displaystyle\sum_{m=1}^{\vert F \vert}e^{-i \sum_{l=s_{j-1}+1}^{s_j}\gamma_l v_m}}{\vert F \vert} .
    \end{equation}
    We can find the upper bound of $\vert \langle f \vert \psi_{p,\mathcal{C},F}(\boldsymbol{\gamma},\boldsymbol{\beta})\rangle \vert^2$ by maximizing $G_p$, as naturally,
    \begin{equation}
        \max_{\boldsymbol{\gamma}\in \mathbb{R}^p ,\boldsymbol{\beta}\in \mathbb{R}^p,\mathcal{C}}\vert \langle f \vert \psi_{p,\mathcal{C},F}(\boldsymbol{\gamma},\boldsymbol{\beta})\rangle \vert^2 \leqslant \max_{\boldsymbol{\gamma}\in \mathbb{R}^p ,\boldsymbol{\beta}\in \mathbb{R}^p, \mathbf{v} \in \mathbb{R}^{\vert F \vert}, z \in \mathbb{R}}G_p(\boldsymbol{\gamma},\boldsymbol{\beta},\mathbf{v},z).
    \end{equation}

    The partial derivative of $G_p(\boldsymbol{\gamma},\boldsymbol{\beta},\mathbf{v},z)$ with respect to $v_m, m \in \{1,2,\cdots,\vert F \vert\}$, is given by,
    \begin{equation}
        \frac{\partial G_p(\boldsymbol{\gamma},\boldsymbol{\beta},\mathbf{v},z)}{\partial v_m} = \frac{2}{\vert F \vert} \operatorname{Re}(\frac{\partial g_p(\boldsymbol{\gamma},\boldsymbol{\beta},\mathbf{v},z)}{\partial v_m}\overline{g_p}(\boldsymbol{\gamma},\boldsymbol{\beta},\mathbf{v},z)),
    \end{equation}
    where $\operatorname{Re}(\cdot)$ returns the real part of the given number, and,
    \begin{align}
              & \frac{\partial g_p(\boldsymbol{\gamma},\boldsymbol{\beta},\mathbf{v},z)}{\partial v_m} \overline{g_p}(\boldsymbol{\gamma},\boldsymbol{\beta},\mathbf{v},z) \nonumber\\
            = & \sum_{\substack{\mathbf{s} \in \lhd_0\binom{[p]}{k},\\ k\in \{1,2,\cdots,p\}}} \left[ \begin{aligned}
                    & e^{-i \sum_{j=s_k+1}^{p}\gamma_j z} \left(\prod_{j=1}^{k}(e^{-i\beta_{s_j}}-1)\right) \times\\
                    & \sum_{j=1}^{k} \frac{-i\left(\sum_{l=s_{j-1}+1}^{s_j}\gamma_l\right)e^{-i\sum_{l=s_{j-1}+1}^{s_j}\gamma_l v_m}}{\vert F \vert} \prod_{\substack{j'=1,\\j'\neq j}}^{k} \frac{\displaystyle\sum_{m'=1}^{\vert F \vert} e^{-i\sum_{l=s_{j'-1}+1}^{s_{j'}}\gamma_l v_{m'}}}{\vert F \vert}
                \end{aligned} \right] \times \nonumber \\
              & \sum_{\substack{\mathbf{s} \in \lhd_0\binom{[p]}{k},\\ k\in \{0,1,\cdots,p\}}} e^{i \sum_{j=s_k+1}^{p}\gamma_j z} \prod_{j=1}^{k}(e^{i\beta_{s_j}}-1) \prod_{j=1}^{k}\frac{\displaystyle\sum_{m'=1}^{\vert F \vert}e^{i \sum_{l=s_{j-1}+1}^{s_j}\gamma_l v_{m'}}}{\vert F \vert} \nonumber\\
              = & -i\!\!\!\!\!\!\sum_{\substack{\mathbf{s} \in \lhd_0\binom{[p]}{k},\\\mathbf{s'} \in \lhd_0\binom{[p]}{k'},\\ k\in \{1,2,\cdots,p\},\\ k'\in \{0,1,\cdots,p\}}} \left[ \begin{aligned} \\
                & \prod_{j=1}^{k}(e^{-i\beta_{s_j}}-1)\prod_{j=1}^{k'}(e^{i\beta_{s'_j}}-1) \left(\sum_{t \in [\vert F \vert]^{k'}} \frac{e^{i\left(\sum_{j=1}^{k'} \sum_{l=s_{j-1}+1}^{s_j}\gamma_l v_{t_j} + \sum_{j=s'_{k'+1}}^{p}\gamma_l z\right)}}{\vert F \vert^{k'}}\right)\\
                & \sum_{j=1}^{k}\sum_{t\in [\vert F \vert]^{j-1}}\sum_{t'\in [\vert F \vert]^{k-j}}\frac{(\sum_{l=s_{j-1}+1}^{s_j}\gamma_l)e^{-i \left( \begin{aligned} 
                    & \sum_{j'=1}^{j-1}\sum_{l=s_{j'-1}+1}^{s_{j'}}\gamma_l v_{t_{j'}} + \sum_{l=s_{j-1}+1}^{s_j}\gamma_l v_m +\\ 
                    & \sum_{j'=j+1}^{k}\sum_{l=s_{j'-1}+1}^{s_{j'}}\gamma_l v_{t'_{j'-j}} +\sum_{l=s_k+1}^{p}\gamma_l z 
                \end{aligned} \right) }}{\vert F \vert^{k}}
            \end{aligned} \right].
    \end{align}
    When $\boldsymbol{\beta} \in \{\boldsymbol{\beta}\mid \beta_i = 2k_i\pi,k_i \in \mathbb{Z}, \forall i \in [p]\}$, the partial derivative, $\frac{\partial G_p(\boldsymbol{\gamma},\boldsymbol{\beta},\mathbf{v},z)}{\partial v_m}=0$. 
    From Lemma~\ref{lemma:umcancel}, all phase separation and mixing operations cancel out, leaving the circuit in the initial state $\vert F \rangle$,
    where the measurement probability of each computational basis state that encoded a possible solution is $\frac{1}{\vert F \vert}$, which is not the maximum.
    Alternatively, when $(\boldsymbol{\gamma},\boldsymbol{\beta},\mathbf{v},z) \in T^{(1)}_{p}$, where
    \begin{align}
        T^{(1)}_{p} \coloneqq \left\{ \boldsymbol{\gamma},\boldsymbol{\beta},\mathbf{v},z \middle\vert \begin{aligned} 
            & (\beta_i \!=\! a_i\pi) \wedge (\gamma_i v_j\!+\!d_i\!=\! b_{i,j}\pi) \wedge (\gamma_i z \! + \!d_i\!=\! c_{i}\pi) \wedge \\
            & (a_i \in \mathbb{Z}, b_{i,j}\in \mathbb{Z}, c_{i}\in \mathbb{Z}, d_i \in \mathbb{R}),  \forall i \in [p], \forall j \in [\vert F \vert]
        \end{aligned} \right\},
    \end{align}
    $\frac{\partial g_p(\boldsymbol{\gamma},\boldsymbol{\beta},\mathbf{v},z)}{\partial v_m} \overline{g_p}(\boldsymbol{\gamma},\boldsymbol{\beta},\mathbf{v},z)$ retains only the imaginary part that $\frac{\partial G_p(\boldsymbol{\gamma},\boldsymbol{\beta},\mathbf{v},z)}{\partial v_m}=0, \forall m \in [\vert F \vert]$.
    The points where $G_p$ attains its maximum value are certainly among $T^{(1)}_{p}$ as,
    \begin{equation}
        \max_{\boldsymbol{\gamma}\in \mathbb{R}^p ,\boldsymbol{\beta}\in \mathbb{R}^p, \mathbf{v} \in \mathbb{R}^{\vert F \vert}, z \in \mathbb{R}} G_p(\boldsymbol{\gamma},\boldsymbol{\beta},\mathbf{v},z) = \max_{(\boldsymbol{\gamma},\boldsymbol{\beta},\mathbf{v},z) \in T^{(1)}_{p}} G_p(\boldsymbol{\gamma},\boldsymbol{\beta},\mathbf{v},z).
    \end{equation}

    Let $(\boldsymbol{\gamma},\boldsymbol{\beta},\mathbf{v},z) \in T^{(1)}_{p}$, then, we get
    \begin{align}
        \begin{aligned}
            G_p(\boldsymbol{\gamma},\boldsymbol{\beta},\mathbf{v},z) = & \frac{1}{\vert F \vert} \left\lvert \sum_{\substack{\mathbf{s} \in \lhd_0\binom{[p]}{k},\\ k\in \{0,1,\cdots,p\}}} e^{-i \sum_{j=s_k+1}^{p}(c_l \pi - d_l)} \prod_{j=1}^{k}(e^{-i a_{s_j} \pi}-1)\prod_{j=1}^{k}\frac{\displaystyle\sum_{m=1}^{\vert F \vert}e^{-i \sum_{l=s_{j-1}+1}^{s_j}(b_{l,m}\pi-d_l)}}{\vert F \vert}  \right\rvert^2 \\
            = & \frac{1}{\vert F \vert} \left\lvert \sum_{\substack{\mathbf{s} \in \lhd_0\binom{[p]}{k},\\ k\in \{0,1,\cdots,p\}}} e^{-i \sum_{j=s_k+1}^{p}c_l \pi} \prod_{j=1}^{k}(e^{-i a_{s_j} \pi}-1)\prod_{j=1}^{k}\frac{\displaystyle\sum_{m=1}^{\vert F \vert}e^{-i \sum_{l=s_{j-1}+1}^{s_j} b_{l,m}\pi}}{\vert F \vert}  \right\rvert^2,
        \end{aligned}\label{eqn:mapping}
    \end{align}
    where $a_i \in \mathbb{Z}, b_{i,j}\in \mathbb{Z}, c_{i}\in \mathbb{Z}, d_i \in \mathbb{R},  \forall i \in [p], \forall j \in [\vert F \vert]$.
    Following this, we can further narrow the search space for finding the maximum of the $G_p$ to the set, $T^{(2)}_p$,
    \begin{align}
        T^{(2)}_{p} \coloneqq \left\{ \boldsymbol{\gamma},\boldsymbol{\beta},\mathbf{v},z \middle\vert \boldsymbol{\gamma} \in \{0,\pi\}^p, \boldsymbol{\beta}\in \{0,\pi\}^p, \mathbf{v} \in \mathbb{Z}^{\vert F \vert} , z \in \mathbb{Z} \right\},
    \end{align}
    such that,
    \begin{equation}
        \max_{(\boldsymbol{\gamma},\boldsymbol{\beta},\mathbf{v},z) \in T^{(1)}_{p}} G_p(\boldsymbol{\gamma},\boldsymbol{\beta},\mathbf{v},z) = \max_{(\boldsymbol{\gamma},\boldsymbol{\beta},\mathbf{v},z) \in T^{(2)}_{p}} G_p(\boldsymbol{\gamma},\boldsymbol{\beta},\mathbf{v},z).
    \end{equation}

    We first consider the case that $\boldsymbol{\gamma} = \pi\mathbf{1}_p , \boldsymbol{\beta} =\pi\mathbf{1}_p$, where $\pi\mathbf{1}_p$ represents a length-$p$ vector with all elements $\pi$, we get,
    \begin{align}
        \begin{aligned}
            & G_p(\boldsymbol{\gamma}=\pi \mathbf{1}_p, \boldsymbol{\beta}=\pi \mathbf{1}_p ,\mathbf{v}\in \mathbb{Z}^{\vert F \vert},z \in \mathbb{Z}) \\
            = & \frac{1}{\vert F \vert} \left( \cos(p \pi z) + \sum_{k=1}^{p} \cos((p-k)\pi z) \sum_{j=0}^{k-1}\sum_{\mathbf{s} \in \lhd_0\binom{[k-1]}{j}\rhd_k}(-2)^{j+1}\prod_{i=1}^{j+1}\frac{\displaystyle\sum_{m=1}^{\vert F \vert}\cos((s_i-s_{i-1})\pi v_m)}{\vert F \vert}\right)^2.
        \end{aligned}
    \end{align}
    Assume there are $n$ odd numbers and $m$ even numbers in $\mathbf{v}$, and let $r \coloneqq \frac{-n+m}{\vert F \vert}$, then, we can define
    \begin{equation}
        H_k(r) = \sum_{j=0}^{k-1}\sum_{\mathbf{s} \in \lhd_0\binom{[k-1]}{j}\rhd_k}(-2)^{j+1}\prod_{i=1}^{j+1}h(s_j,s_{j-1},r),
    \end{equation}
    where
    \begin{align}
        h(s_j,s_{j-1},r) \coloneqq \frac{\displaystyle\sum_{m=1}^{\vert F \vert}\cos((s_j-s_{j-1})\pi v_m)}{\vert F \vert} = \begin{cases}
            1 & s_j-s_{j-1}\text{ is even},\\
            r & s_j-s_{j-1}\text{ is odd}.
        \end{cases}
    \end{align}
    Expand $H_{k+1}(r)$, we get,
    \begin{align}
            & H_{k+1}(r) = \sum_{j=0}^{k}\sum_{\mathbf{s} \in \lhd_0\binom{[k]}{j}\rhd_{k+1}}(-2)^{j+1} \prod_{i=1}^{j+1}h(s_j,s_{j-1},r) \nonumber\\
            = & \!\!\!\!\!\!\sum_{\substack{\mathbf{s} \in \lhd_0\binom{[k-1]}{j}\rhd_k\rhd_{k+1},\\j\in\{0,1,\dots,k-1\}}}\!\!\!\!\!\!\!\!\!\!\!\!\!\!(-2)^{j+2} \prod_{i=1}^{j+2}h(s_j,s_{j-1},r) + \!\!\!\!\!\!\!\!\!\!\!\!\!\! \sum_{\substack{\mathbf{s} \in \lhd_0\binom{[k-2]}{j}\rhd_{k-1}\rhd_{k+1},\\j\in\{0,1,\dots,k-2\}}}\!\!\!\!\!\!\!\!\!\!\!\!\!\!(-2)^{j+2} \prod_{i=1}^{j+2}h(s_j,s_{j-1},r) + \!\!\!\!\!\!\!\!\!\!\!\!\!\!\sum_{\substack{\mathbf{s} \in \lhd_0\binom{[k-2]}{j}\rhd_{k+1},\\ j \in \{0,1,\dots,k-2\}}}\!\!\!\!\!\!\!\!\!\!\!\!\!\!(-2)^{j+1} \prod_{i=1}^{j+1}h(s_j,s_{j-1},r) \nonumber\\
            = & -2h(k+1,k,r)\sum_{j=0}^{k-1}\sum_{\mathbf{s} \in \lhd_0\binom{[k-1]}{j}\rhd_k\rhd_{k+1}}(-2)^{j+1} \prod_{i=1}^{j+1}h(s_j,s_{j-1},r) \nonumber\\
              & -2h(k+1,k-1,r)\sum_{j=0}^{k-2}\sum_{\mathbf{s} \in \lhd_0\binom{[k-2]}{j}\rhd_{k-1}\rhd_{k+1}}(-2)^{j+1} \prod_{i=1}^{j+1}h(s_j,s_{j-1},r) \nonumber\\
              & + \sum_{j=0}^{k-2}\sum_{\mathbf{s} \in \lhd_0\binom{[k-2]}{j}\rhd_{k+1}}(-2)^{j+1} \prod_{i=1}^{j+1}h(s_j,s_{j-1},r) \nonumber\\
            = & -2rH_k(r)-2H_{k-1}(r)+H_{k-1}(r) = -2rH_k(r)-H_{k-1}(r). 
    \end{align}
    Consider $z\in \mathbb{Z}$ as a constant number, we define a function of $r$ as,
    \begin{align}
        \begin{aligned}
            & \mathcal{G}_{p,z}(r) = \frac{1}{\vert F \vert}\left(\cos(p\pi z)+\sum_{k=1}^{p}\cos((p-k)\pi z)H_k(r)\right)^2,\\
            & H_{k+1}(r)=-2rH_k(r)-H_{k-1}(r),\;\;H_1(r)=-2r,\;\;H_2(r)=4r^2-2.
        \end{aligned}
    \end{align}
    Then, we can find the maximal points of $G_p(\boldsymbol{\gamma}=\pi \mathbf{1}_p, \boldsymbol{\beta}=\pi \mathbf{1}_p ,\mathbf{v}\in \mathbb{Z}^{\vert F \vert},z \in \mathbb{Z})$, by solving the problem:
    \begin{align}
        \max_{r}\mathcal{G}_{p,z}(r), \quad s.t.\;r^2-1\leqslant 0.
    \end{align}
    Using the method of Lagrange multiplier, we have the Lagrange function as,
    \begin{align}
        \mathcal{L}(r,\nu)=-\mathcal{G}_{p,z}(r)+\nu(r^2-1),
    \end{align} 
    where $\nu$ is the Lagrange multiplier. Then, the problem is transformed as,
    \begin{align}
        \left\{
             \begin{aligned} 
                & \frac{\partial \mathcal{L}(r,\nu)}{\partial r} = -\frac{2}{\vert F\vert} \left( \cos(p\pi z) + \sum_{k=1}^{p} \cos((p-k)\pi z) H_{k}(r)\right) \left( \sum_{k=1}^{p} \cos((p-k)\pi z) \frac{\partial H_{k}(r)}{\partial r}\right) + 2\nu r = 0  \\
                & r^2-1 \leqslant 0, \\
                & \nu \geqslant 0, \\
                & \nu(r^2-1) = 0.
            \end{aligned}
        \right.
    \end{align}
    In the case of where $z$ is even, we find that if $r=-1$,
    \begin{align}
        \begin{aligned} 
            & H_k(-1)=2, \left.\frac{\partial H_k(r)}{\partial r}\right|_{r=-1}=-2k^2,\quad\forall k \in [p], \\
            & \mathcal{G}_{p,z}(-1)=\frac{(2p+1)^2}{\vert F \vert}, \\
            & \nu = \frac{p(p+1)(2p+1)^2}{3\vert F \vert}>0.
        \end{aligned}
    \end{align}
    In the case of where $z$ is odd, we find that if $r=1$,
    \begin{align}
        \begin{aligned} 
            & H_k(1)=2(-1)^k, \left.\frac{\partial H_k(r)}{\partial r}\right|_{r=-1}=2(-1)^k k^2,\quad\forall k \in [p], \\
            & \mathcal{G}_{p,z}(1)=\frac{(2p+1)^2}{\vert F \vert}, \\
            & \nu = \frac{p(p+1)(2p+1)}{\vert F \vert}>0.
        \end{aligned}
    \end{align}
    Hence, when $z$ is even, $r=-1$, or when $z$ is odd, $r=1$, $\mathcal{G}_{p,z}(r)$ attains its maximum $\frac{(2p+1)^2}{\vert F \vert}$. 
    Following this, we get,
    \begin{align}
        \begin{aligned}
            & \max_{\mathbf{v},z} G_p(\boldsymbol{\gamma}=\pi \mathbf{1}_p, \boldsymbol{\beta}=\pi \mathbf{1}_p ,\mathbf{v}\in \mathbb{Z}^{\vert F \vert},z \in \mathbb{Z}) = \frac{(2p+1)^2}{\vert F \vert},\\
            & \argmax_{\mathbf{v},z} G_p(\boldsymbol{\gamma}=\pi \mathbf{1}_p, \boldsymbol{\beta}=\pi \mathbf{1}_p ,\mathbf{v}\in \mathbb{Z}^{\vert F \vert},z \in \mathbb{Z}) \in T^{(3)}_p,
        \end{aligned}
    \end{align}
    where
    \begin{align}
        T^{(3)}_p \coloneqq \left\{\mathbf{v},z \middle\vert \begin{aligned}
         & (v_i\!=\!2a_i\!+\!1 \wedge z\!=\!2b, a_i\in \mathbb{Z},b\in \mathbb{Z}, \forall i \in [\vert F\vert]) \vee \\
         & (v_i\!=\!2a_i \wedge z\!=\!2b\!+\!1, a_i\in \mathbb{Z},b\in \mathbb{Z},\forall i \in [\vert F\vert]) 
        \end{aligned}\right\}.
    \end{align}

    For the case that the parameters $(\boldsymbol{\gamma},\boldsymbol{\beta}) \in S_m, S_m \coloneqq \bigcup_{\mathbf{s}\in \binom{[p]}{m}}\{\boldsymbol{\gamma}, \boldsymbol{\beta} \mid \beta_i = 0 \vee \gamma_i = 0, \forall i \in \mathbf{s}; \beta_j = \pi \wedge \gamma_j = \pi, \forall j \notin \mathbf{s}\}$, 
    we consider back to the Lemma~\ref{lemma:umcancel}, and Lemma~\ref{lemma:upcancel}, the depth reduces to $p-m$. 
    Then, 
    \begin{align}
        \max_{\mathbf{v},z} G_{p\!-\!m}(\boldsymbol{\gamma}=\pi \mathbf{1}_{p\!-\!m}, \boldsymbol{\beta}=\pi \mathbf{1}_{p\!-\!m} ,\mathbf{v}\in \mathbb{Z}^{\vert F \vert},z \in \mathbb{Z}) = \frac{(2(p-m)+1)^2}{\vert F \vert}<\frac{(2p+1)^2}{\vert F \vert},\; \forall m \in [p],
    \end{align}
    thus,
    \begin{align}
        \max_{\boldsymbol{\gamma}\in \mathbb{R}^p ,\boldsymbol{\beta}\in \mathbb{R}^p, \mathbf{v} \in \mathbb{R}^{\vert F \vert}, z \in \mathbb{R}}G_p(\boldsymbol{\gamma},\boldsymbol{\beta},\mathbf{v},z) = \max_{(\boldsymbol{\gamma},\boldsymbol{\beta},\mathbf{v},z) \in T^{(2)}_{p}} G_p(\boldsymbol{\gamma},\boldsymbol{\beta},\mathbf{v},z) = \frac{(2p+1)^2}{\vert F \vert}.
    \end{align}
    According to Equation~(\ref{eqn:mapping}), we get,
    \begin{align}
        \argmax_{\boldsymbol{\gamma},\boldsymbol{\beta},\mathbf{v},z} G_p(\boldsymbol{\gamma}\in \mathbb{R}^p ,\boldsymbol{\beta}\in \mathbb{R}^p, \mathbf{v} \in \mathbb{R}^{\vert F \vert}, z \in \mathbb{R}) \in T_p,
    \end{align}
    where
    \begin{align}
        T_{p} \coloneqq \left\{ \boldsymbol{\gamma},\boldsymbol{\beta},\mathbf{v},z \middle\vert \begin{aligned} 
            & \left(\beta_i \!=\! (2a_i\!+\!1)\pi, a_i \in \mathbb{Z},\forall i \in [p] \right) \wedge \\
            & \left(\begin{aligned}
                & (\gamma_i v_j \!+\! d_i \!=\! 2 b_{i,j}\pi \wedge \gamma_i z \!+\! d_i \!=\! (2 c_i \!+\! 1)\pi, b_{i,j},c_i,d_i \in \mathbb{Z}, \forall i \in [p], \forall j \in [\vert F \vert]) \vee \\
                & (\gamma_i v_j \!+\! d_i \!=\! (2 b_{i,j}\!+\!1)\pi \wedge \gamma_i z \!+\! d_i \!=\! 2 c_i\pi, b_{i,j},c_i,d_i \in \mathbb{Z}, \forall i \in [p], \forall j \in [\vert F \vert]) 
            \end{aligned}\right)
        \end{aligned} \right\}.
        \label{eqn:maxset}
    \end{align}

    Finally, according to Equation~\ref{eqn:maxset}, when $G_p$ attains its maximum, $z \neq v_i, \forall i \in [\vert F \vert]$, we arrive at the conclusion that the probability of sampling a computational basis state $\vert f \rangle$ from a depth-$p$ GM-QAOA circuit satisfies as,
    \begin{align}
        \vert \langle f \vert \psi_{p,\mathcal{C},F}(\boldsymbol{\gamma},\boldsymbol{\beta})\rangle \vert^2 < \max_{\boldsymbol{\gamma}\in \mathbb{R}^p ,\boldsymbol{\beta}\in \mathbb{R}^p, \mathbf{v} \in \mathbb{R}^{\vert F \vert}, z \in \mathbb{R}}G_p(\boldsymbol{\gamma},\boldsymbol{\beta},\mathbf{v},z) = \frac{(2p+1)^2}{\vert F \vert}.
    \end{align}
\end{proof}

\section{Proof of Theorem 3}
\label{sec:profThe3}
\begin{proof}[Proof of Theorem~\ref{the:3}]
    From Definition~\ref{def:approx} and Definition~\ref{def:mmr}, we get,
    \begin{equation}
        \begin{split}
            & \mu_{\frac{1}{(2p+1)^2}} - \alpha \\
            = & \frac{(2p+1)^2\sum_{i=1}^{\left\lceil \frac{\vert F \vert}{(2p+1)^2} \right\rceil } C(f^{(i)})}{\vert F \vert C(f^{(1)})} - \frac{\sum_{i=1}^{\vert F \vert} \vert \langle f^{(i)} \vert \psi_{p,\mathcal{C},F}(\boldsymbol{\gamma},\boldsymbol{\beta})\rangle \vert^2 C(f^{(i)})}{C(f^{(1)})} \\
            = & \frac{\sum_{i=1}^{\left\lceil \frac{\vert F \vert}{(2p+1)^2} \right\rceil } \left(\frac{(2p+1)^2}{\vert F \vert} - \vert \langle f^{(i)} \vert \psi_{p,\mathcal{C},F}(\boldsymbol{\gamma},\boldsymbol{\beta})\rangle \vert^2 \right)C(f^{(i)})}{C(f^{(1)})} - \frac{\sum_{i=\left\lceil \frac{\vert F \vert}{(2p+1)^2} \right\rceil +1}^{\vert F \vert} \vert \langle f^{(i)} \vert \psi_{p,\mathcal{C},F}(\boldsymbol{\gamma},\boldsymbol{\beta})\rangle \vert^2 C(f^{(i)})}{C(f^{(1)})} .
            \label{eqn:the3proof1}
        \end{split}
    \end{equation}
    From equation~\ref{eqn:sort},
    \begin{equation}
        \begin{split}
            (\ref{eqn:the3proof1}) \geqslant & \frac{\sum_{i=1}^{\left\lceil \frac{\vert F \vert}{(2p+1)^2} \right\rceil } \left(\frac{(2p+1)^2}{\vert F \vert} - \vert \langle f^{(i)} \vert \psi_{p,\mathcal{C},F}(\boldsymbol{\gamma},\boldsymbol{\beta})\rangle \vert^2 \right)C(f^{(i)})}{C(f^{(1)})} \\
            & - \frac{C(f^{(\left\lceil \frac{\vert F \vert}{(2p+1)^2} \right\rceil +1)}) \sum_{i=\left\lceil \frac{\vert F \vert}{(2p+1)^2} \right\rceil +1}^{\vert F \vert} \vert \langle f^{(i)} \vert \psi_{p,\mathcal{C},F}(\boldsymbol{\gamma},\boldsymbol{\beta})\rangle \vert^2 }{C(f^{(1)})} \\
            = & \frac{\sum_{i=1}^{\left\lceil \frac{\vert F \vert}{(2p+1)^2} \right\rceil } \left(\frac{(2p+1)^2}{\vert F \vert} - \vert \langle f^{(i)} \vert \psi_{p,\mathcal{C},F}(\boldsymbol{\gamma},\boldsymbol{\beta})\rangle \vert^2 \right)C(f^{(i)})}{C(f^{(1)})} \\ 
            & - \frac{C(f^{(\left\lceil \frac{\vert F \vert}{(2p+1)^2} \right\rceil +1)}) \left(1 - \sum_{i=1}^{\left\lceil \frac{\vert F \vert}{(2p+1)^2} \right\rceil} \vert \langle f^{(i)} \vert \psi_{p,\mathcal{C},F}(\boldsymbol{\gamma},\boldsymbol{\beta})\rangle \vert^2 \right)}{C(f^{(1)})}.
        \end{split}
    \end{equation}
    As $ \sum_{i=1}^{\left\lceil \frac{\vert F \vert}{(2p+1)^2} \right\rceil} \frac{(2p+1)^2}{\vert F \vert} \geqslant 1$,
    \begin{equation}
        \begin{split}
            (\ref{eqn:the3proof1}) \geqslant & \frac{\sum_{i=1}^{\left\lceil \frac{\vert F \vert}{(2p+1)^2} \right\rceil } \left(\frac{(2p+1)^2}{\vert F \vert} - \vert \langle f^{(i)} \vert \psi_{p,\mathcal{C},F}(\boldsymbol{\gamma},\boldsymbol{\beta})\rangle \vert^2 \right)C(f^{(i)})}{C(f^{(1)})} \\ 
            & - \frac{C(f^{(\left\lceil \frac{\vert F \vert}{(2p+1)^2} \right\rceil +1)}) \sum_{i=1}^{\left\lceil \frac{\vert F \vert}{(2p+1)^2} \right\rceil} \left( \frac{(2p+1)^2}{\vert F \vert} - \vert \langle f^{(i)} \vert \psi_{p,\mathcal{C},F}(\boldsymbol{\gamma},\boldsymbol{\beta})\rangle \vert^2 \right)}{C(f^{(1)})} \\
            = & \frac{\sum_{i=1}^{\left\lceil \frac{\vert F \vert}{(2p+1)^2} \right\rceil } \left(\frac{(2p+1)^2}{\vert F \vert} - \vert \langle f^{(i)} \vert \psi_{p,\mathcal{C},F}(\boldsymbol{\gamma},\boldsymbol{\beta})\rangle \vert^2 \right) \left(C(f^{(i)})-C(f^{(\left\lceil \frac{\vert F \vert}{(2p+1)^2} \right\rceil +1)})\right)}{C(f^{(1)})}
        \end{split}
    \end{equation}
    From Theorem~\ref{the:1} that $\frac{(2p+1)^2}{\vert F \vert} > \vert \langle f^{(i)} \vert \psi_{p,\mathcal{C},F}(\boldsymbol{\gamma},\boldsymbol{\beta})\rangle \vert^2$, we finally get,
    \begin{equation}
        (\ref{eqn:the3proof1}) \geqslant 0.
    \end{equation}
\end{proof}

\section{Problem Definition and Instance Sets}
\label{sec:problemset}
\subsection{Traveling Salesman Problem}
The traveling salesman problem seeks to find the shortest possible route that visits each city once and returns to the original city. Here, following \cite{lucas2014ising}, we fix the first city, and formulate an $n$-city traveling salesman problem instance $(C,F)$ as follows,
\begin{align}
    \begin{aligned}
        & \min_{\mathbf{x}\in F} C(\mathbf{x}); \\
        & C(\mathbf{x}) = \sum_{i=1}^{n-1}w_{0,i}x_{1,i}+\sum_{i=1}^{n-1}w_{i,0}x_{n\!-\!1,i}+\sum_{t=2}^{n-1}\sum_{\substack{i,j=1,\\i\neq j}}w_{i,j}x_{t\!-\!1,i}x_{t,j};\\
        & F =  \left\{\mathbf{x} \middle\vert \mathbf{x}\in \{0,1\}^{(n-1)^2};\; \forall t \in [n\!-\!1], \sum_{t=1}^{n-1}x_{t,i}=1;\; \forall i \in [n\!-\!1], \sum_{i=1}^{n-1}x_{t,i}=1 \right\},
    \end{aligned}
\end{align}
where $w_{i,j} \in \mathbb{R}$ represents the distance between city $i$ and city $j$.

Our instance sets include problems ranging from $7$-city to $13$-city, and $48$ instances for each size. The coordinates of the cities are sampled from a uniform distribution within interval $[0,1]$.

\subsection{Max-$k$-Colorable-Subgraph Problem}
The max-$k$-colorable-subgraph problem seeks to find a subgraph with the maximum number of edges that can be properly colored using at most $k$ colors, such that no two adjacent vertices share the same color. Here, we use the definition in \cite{wang2020x}, given an $n$-vertex undirected graph with the edge set $E$, the max-$k$-colorable-subgraph problem instance $(C,F)$ is formulated as follows,
\begin{align}
    \begin{aligned}
        & \max_{\mathbf{x}\in F} C(\mathbf{x}); \\
        & C(\mathbf{x}) = \vert E \vert - \sum_{v=1}^{k}\sum_{(i,j)\in E} x_{i,v}x_{j,v};\\
        & F =  \left\{\mathbf{x} \middle\vert \mathbf{x}\in \{0,1\}^{n\times k};\; \forall i \in [n], \sum_{v=1}^{k}x_{i,v}=1 \right\}.
    \end{aligned}\label{eqn:mkcs}
\end{align}

Our instance sets include problems ranging from graphs with $11$ to $17$ vertices, and the color number $k$ is set to $3$. We generate $48$ instances for each vertex number. The graphs are generated randomly, but we specifically select graphs where the solution is the given graph itself, which increases the difficulty of the problem.

\subsection{Max-Cut Problem}
The max-cut problem is equivalent to the max-$k$-colorable-subgraph problem with $k=2$. Since there are only two colors, we can encode the color using a single bit. Given an $n$-vertex undirected graph with the edge set $E$, the max-cut problem instance $(C,F)$ is formulated as follows,
\begin{align}
    \begin{aligned}
        & \max_{\mathbf{x}\in F} C(\mathbf{x}); \\
        & C(\mathbf{x}) = \sum_{(i,j)\in E}(x_i - x_j)^2;\\
        & F =  \{0,1\}^{n}.
    \end{aligned}
\end{align}
It is worth noting that whether using this formulation or the max-$k$-colorable-subgraph formulation (i.e., Equation~\ref{eqn:mkcs}) with $k=2$, the resulting objective value distribution remains the same.

Our instance sets include problems with $3$-regular graphs having $16$, $18$, $20$, $22$, $24$, $26$, and $28$ vertices. We randomly generate $48$ instances for each vertex number.

\subsection{Max-$k$-Vertex-Cover Problem}
The max-$k$-vertex-cover problem aims to find a subset of vertices with a size of $k$ in an undirected graph, such that the number of edges incident to the vertices in the subset is maximized. Here, we use the definition in \cite{Bartschi_2020}, given an $n$-vertex undirected graph with the edge set $E$, the max-$k$-vertex-cover problem instance $(C,F)$ is formulated as follows,
\begin{align}
    \begin{aligned}
        & \max_{\mathbf{x}\in F} C(\mathbf{x}); \\
        & C(\mathbf{x}) = \vert E \vert - \sum_{(i,j)\in E}(1-x_i)(1-x_j);\\
        & F =  \left\{\mathbf{x} \middle\vert \mathbf{x}\in \{0,1\}^{n};\; \sum_{i=1}^{n}x_{i}=k \right\}.
    \end{aligned}
\end{align}

Our instance sets include problems ranging from graphs with vertex number, $n \in \{18,20,\dots,30\}$, and $k$ is set to $\frac{n}{2}$. We generate $48$ instances for $n$. The graphs are generated using the Erdős–Rényi model with an edge probability of $0.5$.

\end{appendices}

\bibliography{ref}   
\bibliographystyle{plainnat}   

\end{document}